\documentclass{article}

\usepackage{microtype}
\usepackage{graphicx}
\usepackage{subcaption}
\usepackage{booktabs} % for professional tables

\usepackage{hyperref}

\usepackage[preprint]{icml2026}

\usepackage{amsmath}
\usepackage{amssymb}
\usepackage{mathtools}
\usepackage{amsthm}

\usepackage[capitalize,noabbrev]{cleveref}

\theoremstyle{plain}
\newtheorem{theorem}{Theorem}[section]
\newtheorem{proposition}[theorem]{Proposition}

\theoremstyle{definition}

\theoremstyle{remark}

\usepackage[textsize=tiny]{todonotes}

\usepackage{multirow}
\usepackage{algorithm}
\usepackage{algorithmic}
\definecolor{codecomment}{RGB}{63,127,127}
\renewcommand{\algorithmiccomment}[1]{\textcolor{codecomment}{\#\ #1}}

\usepackage{tikz}
\usetikzlibrary{arrows.meta, positioning}

\usepackage{lscape}
\usepackage{caption} % for \captionof
\usepackage[normalem]{ulem}
\usepackage{wrapfig}

\newcommand{\SthreeRec}{S\textsuperscript{3}-Rec}

\icmltitlerunning{ReSID: information theoretic redesign of semantic tokenization for generative recommenders}

\begin{document}

\twocolumn[
  \icmltitle{Rethinking Generative Recommender Tokenizer:\\ Recsys-Native Encoding and Semantic Quantization Beyond LLMs}

  \icmlsetsymbol{equal}{*}

  \begin{icmlauthorlist}
    \icmlauthor{Yu Liang}{equal,csu}
    \icmlauthor{Zhongjin Zhang}{equal,csu}
    \icmlauthor{Yuxuan Zhu}{shopee-sh}
    \icmlauthor{Kerui Zhang}{shopee-sh}
    \icmlauthor{Zhiluohan Guo}{shopee-sh}
    \icmlauthor{Wenhang Zhou}{shopee-sh}
    \icmlauthor{Zonqi Yang}{shopee-sh}
    \icmlauthor{Kangle Wu}{shopee-sh}
    \icmlauthor{Yabo Ni}{ntu}
    \icmlauthor{Anxiang Zeng}{ntu}
    \icmlauthor{Cong Fu}{shopee-sg}
    \icmlauthor{Jianxin Wang}{csu}
    \icmlauthor{Jiazhi Xia}{csu}
  \end{icmlauthorlist}

  \icmlaffiliation{csu}{Central South University, Changsha, China}
  \icmlaffiliation{shopee-sh}{Shopee Pte. Ltd., Shanghai, China}
  \icmlaffiliation{shopee-sg}{Shopee Pte. Ltd., Singapore, Singapore}
  \icmlaffiliation{ntu}{Nanyang Technological University, Singapore, Singapore}

  \icmlcorrespondingauthor{Jiazhi Xia}{xiajiazhi@csu.edu.cn}

  \icmlkeywords{Representation Learning,Discrete Representations,Information-Theoretic Learning,Masked Autoencoders,Quantization,Generative Recommender; Semantic IDs; Recommendation Systems}

  \vskip 0.3in
]

% \printAffiliationsAndNotice{}
\printAffiliationsAndNotice{\icmlEqualContribution}
\begin{abstract}
Semantic ID (SID)-based recommendation is a promising paradigm for scaling sequential recommender systems, but existing methods largely follow a semantic-centric pipeline: item embeddings are learned from foundation models and discretized using generic quantization schemes. This design is misaligned with generative recommendation objectives: semantic embeddings are weakly coupled with collaborative prediction, and generic quantization is inefficient at reducing sequential uncertainty for autoregressive modeling. To address these, we propose \textbf{ReSID}, a recommendation-native, principled SID framework that rethinks representation learning and quantization from the perspective of information preservation and sequential predictability, \textbf{without} relying on LLMs. ReSID consists of two components: (i) \textbf{Field-Aware Masked Auto-Encoding (FAMAE)}, which learns predictive-sufficient item representations from structured features, and (ii) \textbf{Globally Aligned Orthogonal Quantization (GAOQ)}, which produces compact and predictable SID sequences by jointly reducing semantic ambiguity and prefix-conditional uncertainty. Theoretical analysis and extensive experiments across \textbf{ten} datasets show the effectiveness of ReSID. ReSID consistently outperforms strong sequential and SID-based generative baselines by an average of over \textbf{10\%}, while reducing tokenization cost by up to \textbf{122$\times$}. Code is available at \url{https://github.com/FuCongResearchSquad/ReSID}.
\end{abstract}
\section{Introduction}
Generative recommendation based on Semantic IDs (SIDs) has emerged as a promising approach for scaling recommender systems beyond conventional item-ID modeling~\cite{onerec}. The key idea is to encode each item as a compact sequence of discrete tokens (e.g., $[21,3,54]$), enabling autoregressive decoding, token-by-token, instead of predicting billions of atomic, unrelated item-IDs.

\begin{figure}[t]
    \centering
    \includegraphics[width=1\linewidth]{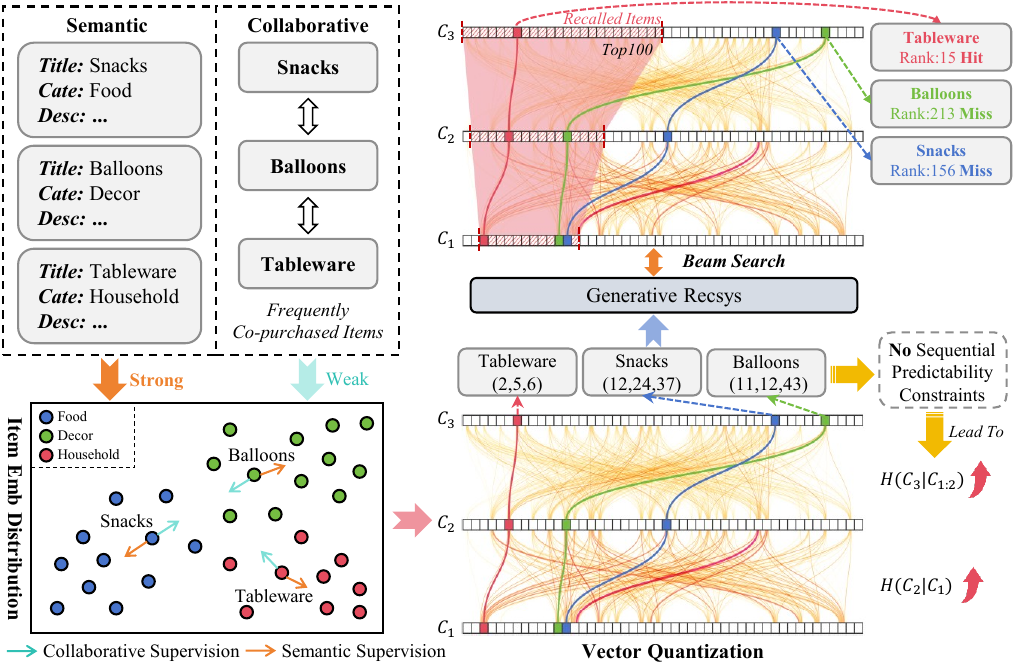}
    \vspace{-2mm}
    \caption{
    Illustration of a traditional semantic-centric SID-based generative recommendation pipeline. Item representations learned from foundation models are weakly aligned with collaborative prediction, and subsequent quantization does not account for sequential predictability in SID decoding, leading to high decoding uncertainty and suboptimal recommendation performance.
    }
    \label{fig:intro_motivation}
\end{figure}

Most existing SID pipelines follow a \emph{semantic-centric} design~\cite{tiger,letter,eager,unger}. Items are first embedded using (M)LLMs, then discretized via vector quantization (e.g., RQ-VAE~\cite{rqvae} or Hierarchical K-Means~\cite{vocabularyTree}), and finally used as tokens in generative recommenders. While effective for vocabulary compression, such pipelines introduce a fundamental \emph{mismatch between SID tokenization and downstream recommender modeling}, manifesting in two core limitations (illustrated in Figure~\ref{fig:intro_motivation}).

\textbf{(1) Misalignment in representation extraction.}
Foundation models are primarily optimized for semantic similarity, which often conflicts with collaborative signals: Items that frequently co-occur in user behaviors (e.g., snacks and balloons for parties) may be far apart in semantic or visual properties. Even with collaboration-oriented fine-tuning~\cite{lcrec,letter,eager,unger}, semantic and collaborative objectives inherently impose competing geometric constraints. In practice, this leads to a compromised embedding structure that is neither semantically clean nor optimally aligned with recommender objectives.

Further, continuously injecting collaborative signals into large semantic encoders incurs substantial training costs~\cite{onerec,onesearch}, and there is \textbf{no principled way to assess} whether the learned embeddings are suitable for recommendation-oriented SID tokenization, since they are optimized indirectly via semantic objectives.

\textbf{(2) Existing quantization methods weaken sequential predictability.}
Existing methods typically emphasize either reconstruction fidelity or hierarchical structures, but fail to jointly consider both. In hierarchical encoding pipelines~\cite{eager,unger,seater}, child indices are often assigned \emph{locally}, \emph{arbitrarily}, and \emph{independently} under each parent. For example, items with substantially different semantics may share the same second-level token ``1'' in codes $(2,1,5)$ and $(9,1,7)$. Since the embedding associated with token ``1'' is shared across all items assigned to that index, this leads to multimodality and high semantic ambiguity, thus lowering the reconstruction fidelity (see theoretical analysis in Section~\ref{sec:method_gaoq}). Conversely, reconstruction-driven quantizers~\cite{tiger,etegrec,forge,onesearch} such as RQ-VAE or RQ-Kmeans optimize the reconstruction error but are agnostic to sequential index dependency across code levels. As a result, the two design choices produce SID sequences that either suffer from high reconstruction error (information loss) or are unfriendly to autoregressive modeling, degrading downstream generative recommendation.

Diverging from the \emph{semantic-centric} paradigm, we rethink representation learning and quantization from an \emph{information-theoretic} perspective. In the representation learning stage, ReSID maximizes the mutual information between a collaboratively sequential-contextualized representation and the semantic-rich target item features, ensuring that the task-relevant information is preserved for downstream recommendation. In the quantization stage, ReSID jointly minimizes reconstruction entropy of the discrete codes and prefix-conditional entropy along the SID sequence using non-parameterized methods, thereby improving predictability in SID-based generative recommendation while reducing computational costs. Our contributions are summarized as follows:

\emph{(1) Recommendation-native representation learning.}
We introduce \emph{Field-Aware Masked Auto-Encoding (FAMAE)}, which learns item embeddings by predicting masked structured features conditioned on the user history. Guided by our information-theoretic analysis, FAMAE preserves recommendation-sufficient information for SIDs and enables intrinsic, task-aware metrics for embedding quality.

\emph{(2) Objective-aligned SID quantization.}
We propose \emph{Globally Aligned Orthogonal Quantization (GAOQ)}, which jointly reduces reconstruction errors and prefix-conditional uncertainty in SIDs. By enforcing globally consistent indexing across hierarchical levels, GAOQ produces compact codes that improve predictability in sequential decoding.

\emph{(3) Strong and consistent empirical results.}
Across \textbf{ten} public datasets, ReSID consistently outperforms strong sequential recommenders and SOTA SID-based generative models, achieving over 10\% relative improvement while reducing tokenization costs by up to 122$\times$ on million-scale datasets. These results show that effective SID construction does not require heavy foundation models, enabling a scalable and adaptable solution for large-scale systems.

\begin{figure*}[t]
    \centering
    \includegraphics[width=0.9\textwidth]{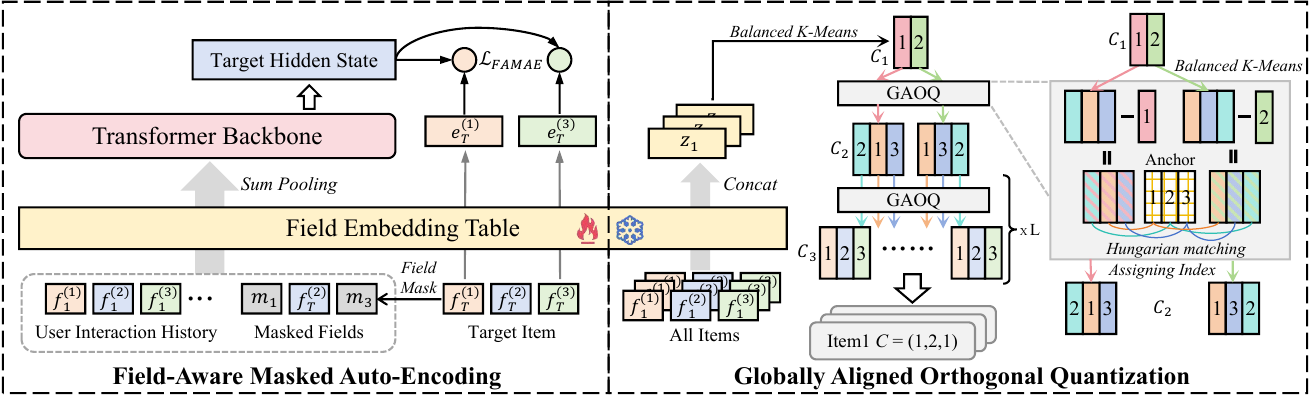}
    \caption{
    Overview of ReSID. FAMAE learns recommendation-sufficient field-level item representations via masked field prediction, and GAOQ discretizes them into compact, autoregressive-decoding-friendly SIDs via global alignment.
    }
    \label{fig:method_resid}
\end{figure*}

\section{Rethinking SID Pipeline} \label{sec:rethinking_sid_pipeline}

\subsection{Problem Definition and Notations}
We consider a sequential recommendation setting that predicts the target item $i_T$ given the user's history $H=(i_1,\dots,i_{T-1})$, where the index $t$ denotes the chronological order of interactions. Each item at position $t$ is associated with a set of raw metadata $X_t$ (e.g., text, images, or other unstructured contents) and a set of structured feature fields $F_t=\{f_t^{(1)},\dots,f_t^{(J)}\}$ engineered from the interaction data and metadata $X_t$, where $f_t^{(k)}$ denotes the $k$-th feature field of $i_t$. A SID-based generative recommender extends this setting by replacing $i_t$'s identifier with a finite-length SID sequence $C_{t}=(c_1,\dots,c_L)$. It mainly includes three components:
(i) an encoder $E_{\theta}$, parameterized by $\theta$, which maps $X_t$ of item $i_t$ to a continuous representation $\mathbf{z}_t$; (ii) a quantizer $Q$, which discretizes $\mathbf{z}_t$ into $C_{t}$; and (iii) a generative model $G_{\phi}$, parameterized by $\phi$, which predicts the target SID, $Y=C_T$, conditioned on the user history.

From an end-to-end perspective, the ideal learning objective can be written as:
\begin{equation}
    \min_{\theta,\phi,Q}\mathbb{E}_{H,X,Y}\big[\mathcal{L}( Y,(G_{\phi} \circ Q \circ E_{\theta})( (X_t)_{t=1}^{T-1} )) \big],
\end{equation}
where $\mathcal{L}(\cdot)$ denotes the cross-entropy loss.

\subsection{The Three-Stage SID Pipeline}
A fundamental challenge in SID-based generative recommendation lies in the \emph{self-referential nature of supervision}. The SIDs produced by the quantizer $Q$ are used as training targets $Y$ for the generative model $G_{\phi}$. Consequently, the quality and semantic consistency of the supervision signal are entirely determined by the upstream encoder–quantizer pair. If the generated SIDs are noisy, misaligned with collaborative signals, or semantically inconsistent, the generative model is forced to learn from a distorted target distribution, with no mechanism for downstream correction.

As a result, existing methods adopt a three-stage pipeline:
(i) representation learning (\textbf{E-stage}),
(ii) discretization into SIDs (\textbf{Q-stage}),
and (iii) autoregressive modeling on SID sequences (\textbf{G-stage}). Information is inevitably compressed across these stages; therefore, the effectiveness of SID-based generation critically depends on whether the intermediate representations and SID codes preserve information that is predictive of downstream generative recommendation.

\subsection{Design Principles for Effective SID Tokenization}
Building on the above discussions, we argue that an effective SID pipeline should follow two principled design criteria.

\textbf{First}, the E-stage should be \emph{collaboration-dominant}: representations must be primarily shaped by user interaction signals, with semantic information serving only as auxiliary contexts. Since discretization is inherently information-losing, ensuring that task-relevant collaborative information dominates the representation is critical to prevent it from being blurred or discarded in the Q-stage. Accordingly, representation quality should be evaluated using \textbf{task-aware metrics} that monitor the preservation of both \emph{recommendation-relevant information} and \emph{semantic fidelity}. 

\textbf{Second}, the Q-stage should preserve as much task-relevant information as possible while explicitly encouraging a sequentially predictable structure in the resulting code sequences. In particular, SID quantization should not only minimize \emph{reconstruction distortion}, but also reduce intrinsic \emph{uncertainty in autoregressive decoding} by promoting sequentially predictable and semantically stable codes. 

Together, these principles align the representation learning and quantization with the information requirements of downstream generative recommendation.

\section{Methodology}
We present ReSID, a recommendation-native SID framework that redesigns both representation learning and quantization from an \emph{information-theoretic} perspective. ReSID consists of two components: \emph{(1) Field-Aware Masked Auto-Encoding (FAMAE)} for learning collaboration-dominant item representations, and \emph{(2) Globally Aligned Orthogonal Quantization (GAOQ)} for constructing compact and autoregressive-decoding-friendly SIDs. The overall pipeline is illustrated in Figure~\ref{fig:method_resid}.

\subsection{Field-Aware Masked Auto-Encoding}

\textbf{Motivation.}
The goal of the E-stage is to learn item representations that are \emph{recommendation-sufficient}, with semantic similarity playing a secondary role. Most existing SID approaches~\cite{onerec,onesearch} ground structured, feature-engineered information into texts or multimodal inputs and then extract embeddings using foundation models. This semantic grounding dilutes task-relevant collaborative signals and makes them particularly vulnerable to information loss during subsequent quantization.

From a recommendation-native perspective, structured information should instead be grounded directly in its native symbols: the engineered feature fields. This design follows a \emph{standard modeling assumption} widely adopted in recommender systems, namely that the prediction target is conditionally independent of raw item metadata given sufficient structured features and user contexts: $Y \perp\!\!\!\perp X \mid (F_T , H)$. This assumption underlies most feature-based recommendation formulations, where long-established feature engineering and interaction modeling are treated as sufficient statistics for predicting user behaviors. FAMAE explicitly builds on this premise to align representation learning with the information requirements of the downstream generative recommendation.

\textbf{FAMAE Objective.}
FAMAE trains a Transformer encoder using a field-aware masked prediction objective. Given a target item at position $T$, a random subset of its feature fields is masked, and the model is trained to predict the masked fields conditioned on the remaining fields and the user history. Formally, the training objective is:
\[
\mathcal{L}_{\mathrm{FAMAE}}(\theta) = \mathbb{E}_{\mathcal{M} \sim \pi}\Big[\sum_{k\in\mathcal{M}} \alpha_k \cdot \Big(- \log q_{\theta,k}(f_T^{(k)} \mid \mathbf{h}_{T})\Big)\Big],
\]
where $\mathcal{M}$ is the set of masked fields sampled from masking policy $\pi$, and $\mathbf{h}_T=g_{\theta}((F_{i})_{i=1}^{T-1},\{f_T^{(j)}\}_{j\notin\mathcal{M}})$ is the contextualized latent representation produced by the Transformer encoder $g_{\theta}$ at position $T$, and $\alpha_k$ controls the relative importance of different feature fields. 
The term $q_{\theta,k}(\cdot \mid \mathbf{h})$ denotes the model's predictive distribution of a field given hidden state $\mathbf{h}$, and is given as:
\begin{equation*}
    q_{\theta,k}(f_T^{(k)} \mid \mathbf{h}) =
    \frac{\exp\!\left(\mathcal{K}(\mathbf{h}, \mathbf{e}_{T}^{(k)})\right)}
    {\sum_{v \in \mathcal{V}_k} \exp\!\left(\mathcal{K}(\mathbf{h}, \mathbf{e}_{v}^{(k)})\right)},
\end{equation*}
where $\mathbf{e}_{T}^{(k)}$ denotes the embedding for field $k$ after embedding lookup $\mathbf{e}_{T}^{(k)}=emb_{\theta}(f_T^{(k)})$,
$\mathcal{K}$ denotes scaled cosine similarity, $\mathcal{K}(\cdot,\cdot)=\sqrt{d}\cdot cos(\cdot,\cdot)$,
and $\mathcal{V}_k$ is the vocabulary of field $k$.
Note that with a mild abuse of notation, we use $\theta$ as the whole parameters of the FAMAE encoder which consists of embedding layer $emb_{\theta}$ and Transformer $g_{\theta}$.

Unlike traditional semantic auto-encoding, this objective enforces separable, field-level supervision, encouraging the representation to retain \emph{fine-grained and task-relevant spatial structures} that are robust under the discretization.

\textbf{Information-Theoretic Interpretation.}
Although FAMAE is implemented as a supervised masked prediction loss, it admits a natural information-theoretic interpretation. Minimizing the FAMAE loss increases a variational lower bound on the mutual information between the learned representation and $F_T$. Formally, we have the following proposition:
\begin{proposition}[Predictive Sufficiency Proxy] \label{prop:sufficiency}
    Consider a masking policy $\pi$ and field weights $\alpha_k \ge 0$. Let $w_k=\alpha_k \Pr_{\mathcal{M} \sim \pi}(k\in \mathcal{M})$, then
    \begin{equation*}
        \sum_{k=1}^J w_k I(\mathbf{h}_T;f_T^{(k)}) \ge \sum_{k=1}^J w_k H(f_T^{(k)}) - \mathcal{L}_{\mathrm{FAMAE}}(\theta).
    \end{equation*}
    In particular, minimizing $\mathcal{L}_{\mathrm{FAMAE}}$ increases the right-hand side, which is a variational lower bound on the mask-weighted mutual information between bottleneck $\mathbf{h}_T$ and the target item's features.
\end{proposition}

This proposition has two main implications for FAMAE.

\textbf{(1) Predictive Sufficiency.}
In FAMAE optimization, the learned representation $\mathbf{h}_T$ compresses information from $F_T$ and $H$, serving as the sufficient statistic if it is used by the downstream models to predict any target $Y$, given $Y\perp\!\!\!\perp X \mid (F_T , H)$. In other words, FAMAE provides a principled proxy for learning representations that preserve the task-relevant information necessary for downstream recommendation, \emph{independent of the design of the Q-stage}.

\textbf{(2) Predictive Superiority.}
Prior sequential recommenders (e.g., SASRec~\cite{sasrec}) are trained with a single-label objective that predicts the fused representation of the next item's features, which can be viewed as a deterministic coarsening $\mathbf{u}_T=f(F_T)$ of the underlying structured feature vectors. The \emph{Data Processing Inequality} implies $I(\mathbf{h}_T;\mathbf{u}_T)\le I(\mathbf{h}_T;F_T)=\sum_{k=1}^J I(\mathbf{h}_T;f_T^{(k)} \mid F_T^{(<k)})$. Equality holds only if $\mathbf{u}_T$ is a sufficient statistic of $F_T$ for predicting $\mathbf{h}_T$, a condition that rarely holds when heterogeneous item fields are fused with non-invertible operators (e.g., pooling or MLP fusion), which discard field identity and fine-grained structures.

Moreover, \emph{mutual predictability} among multiple structured fields naturally retains task-relevant semantic structures (e.g., category hierarchy and attribute constraints), which complements collaborative signals with \textbf{necessary semantic structures} rather than competing with them. This makes FAMAE particularly suitable as a recommendation-native pretraining objective for the SID construction.

\textbf{Item Representation Extraction.}
Notably, the contextual representation $\mathbf{h}_T$ is not suitable for the SID quantization, as it entangles item information with the user history. Since $\mathbf{h}_T$ is a deterministic function of the field-level item embeddings $\mathbf{e}_T$ and the user history $H$, the \emph{Data Processing Inequality} implies $I(\mathbf{e}_T;F_T) \ge I(\mathbf{h}_T;F_T)$, under the data-processing chain $(F_T\rightarrow \mathbf{e}_T \rightarrow \mathbf{h}_T)$. This indicates that $\mathbf{e}_T$ preserves at least as much information about the target item’s structured features $F_T$ as $\mathbf{h}_T$, while remaining independent of user-specific contexts and retaining sufficient task-relevant information for predicting the downstream target $Y$. $\mathbf{e}_T$ provides a more appropriate, task-sufficient basis for the SID construction.

\textbf{Implementations.}
The encoder takes a sequence of items as input, where the first $T-1$ items correspond to the user’s interaction history and the last item serves as the prediction target. For each position, embeddings of all structured feature fields (including item-ID and side information), together with a learnable positional encoding $\mathbf{p}$, are aggregated via sum pooling to form the input token representation. The resulting sequence is then fed into a Transformer backbone with the bidirectional self-attention. We adopt a two-level random masking strategy for the target item. First, we uniformly sample an integer $K\sim U\{1,\dots,J\}$. Second, we randomly select a subset of $K$ fields $\mathcal{M}_K$ to be masked.

To preserve field identity during prediction, we introduce \emph{field-specific} learnable mask tokens $\{\mathbf{m}_1,\dots,\mathbf{m}_J\}$, where each $\mathbf{m}_j$ corresponds to a field $f^{(j)}$. For a masked field, its embedding is replaced by the corresponding mask token. 

As a result, the transformer input of the target item is constructed by:
\begin{equation*}
    \tilde{\mathbf{e}}_T=\mathbf{p}_T+\sum_{j=1}^{J}\mathbf{e}_T^{(j)},
    \quad \mathbf{e}_T^{(j)}=
    \begin{cases}
        \mathbf{m}_j, & j\in \mathcal{M}_K,\\
        emb_{\theta}(f_T^{(j)}), & j \notin \mathcal{M}_K.
    \end{cases}
\end{equation*}

We construct the final representation for each item by concatenating the embeddings of all its corresponding feature fields learned by FAMAE $\operatorname{concat}(\{\mathbf{e}^{(j)}\}_{j=1}^{J})$, which serves as the input to the subsequent SID quantization stage.

\subsection{Task-Aware Metrics for Embedding Quality} \label{sec:metrics}
Existing SID pipelines lack principled, task-aware metrics for evaluating the quality of item representations learned in the E-stage, independent of the downstream Q-stage and G-stage. To address this gap, we introduce two complementary proxy metrics that measure how well the FAMAE embeddings preserve different types of information required.

\textbf{Metric 1: Collaborative Modeling Capability.}
We measure target item prediction accuracy under \emph{full-field masking}, where all structured feature fields of the target item are masked and prediction relies solely on the user history and learned embedding space. From an information-theoretic perspective, this metric evaluates whether the learned representations are sufficient to mediate the conditional dependence between the user history $H$ and target item $F_T$, i.e., whether the predictive information required to infer $F_T$ from $H$ is preserved and accessible through the embeddings.

\textbf{Metric 2: Discriminative Semantics and Spatial Structure.}
We measure item-ID prediction accuracy under \emph{single-field masking}, where only the item-ID field is masked. This metric evaluates whether the structured feature embeddings contain sufficient discriminative semantic information—such as category hierarchy and attribute entailment—to distinguish individual items. It serves as a proxy for assessing whether the learned representation space preserves fine-grained semantic structures required for constructing informative and stable Semantic IDs.

Together, these metrics capture complementary aspects of the embedding quality: collaborative predictability and discriminative semantic structures. Empirically (Figure~\ref{fig:exp_metric_corr}), representations that perform well on both metrics consistently yield higher-quality Semantic IDs and improved downstream generative recommendation performance.

\subsection{Globally Aligned Orthogonal Quantization} \label{sec:method_gaoq}

\textbf{Ideal Objective for SID Quantization.}
Although the downstream generator predicts $c_l$ conditioned on $(C_{(<l)}, H)$, user history $H$ is not accessible in Q-stage. Therefore, the quantizer should be designed to produce codes that are intrinsically autoregressive-decoding-friendly, independent of $H$. From an information-theoretic perspective, an effective SID quantizer should satisfy three desiderata:
(i) low global reconstruction distortion under the discrete bottleneck,
(ii) high semantic contribution from each individual code, and
(iii) low prefix-conditional uncertainty to facilitate sequential prediction. These requirements can be formalized as:
\begin{align} \label{equation:obj}
    \min_Q  H(\mathbf{z} \mid C)
    &+ \mu \sum_l H(\mathbf{z} \mid c_l) + \lambda \sum_l H(c_l \mid C_{(<l)})
    , \notag \\
    &\text{s.t. } H(c_l)\approx \log |c_l|.
\end{align}
where $H(\mathbf{z} \mid C)$ measures the overall reconstruction uncertainty from the full SID sequence, $H(\mathbf{z} \mid c_l)$ forces each code to be individually informative and prefix-invariant, and $H(c_l \mid C_{(<l)})$ captures the intrinsic branching uncertainty faced by an autoregressive decoder, serving as an upper bound on the decoding entropy of the downstream recommender: $H(c_l \mid C_{(<l)},H)\le H(c_l \mid C_{(<l)})$. The entropy constraint ensures uniformly distributed marginal codes for each layer $l$, preventing index collapse.

We next examine how existing quantization schemes violate these principles, before introducing GAOQ.

\textbf{Misalignment of RQ-style Quantization.}
RQ-VAE and RQ-Kmeans only optimize overall reconstruction loss under a discrete bottleneck, often with entropy regularization to encourage balanced code usage. However, their objectives are agnostic to the autoregressive nature of SID decoding. In particular, standard residual quantization assigns codes independently across levels and does not explicitly constrain prefix-conditional uncertainty $H(c_l \mid C_{(<l)})$, nor does it encourage individual codes to carry meaningful semantic information as measured by $H(\mathbf{z} \mid c_l)$. As a result, such methods can achieve low overall reconstruction distortion and balanced marginal usage, yet still produce code sequences that are unstable and difficult to predict sequentially, degrading downstream autoregressive SID modeling.

\textbf{Limitations of Hierarchical K-Means with Local Indexing.} Hierarchical K-Means introduces a tree-structured code by path-dependent branching, which can partially reduce prefix-conditional uncertainty $H(c_l \mid C_{(<l)})$. However, existing hierarchical schemes typically assign child indices \emph{independently} and \emph{locally} within each parent node. As a result, the same code index may correspond to different semantic directions under different prefixes (Figure~\ref{fig:appx_compare}). 

This issue can be formalized via the decomposition below:
\begin{equation} \label{equation:recon-entropy}
    H(\mathbf{z} \mid c_l)=H(\mathbf{z} \mid c_l,C_{(<l)}) + I(\mathbf{z};C_{(<l)} \mid c_l).
\end{equation}
Local indexing increases the prefix-dependent ambiguity term $I(\mathbf{z};C_{(<l)} \mid c_l)$, since the same index at level $l$ can correspond to different semantic directions under different prefixes. This means that the semantic interpretation of a code \emph{depends heavily on its prefix}. Though hierarchical refinement typically reduces the term $H(\mathbf{z} \mid c_l,C_{(<l)})$, the increase in $I(\mathbf{z};C_{(<l)} \mid c_l)$ can offset this gain, leading to larger $H(\mathbf{z} \mid c_l)$. Consequently, individual codes become less informative and less prefix-invariant, which harms both code-level semantic stability and downstream decoding.

\textbf{GAOQ Design.} GAOQ is designed to jointly optimize all three objectives in Eqn.~\eqref{equation:obj} by combining hierarchical vector quantization with globally aligned indexing.

First, GAOQ reduces prefix-conditional uncertainty $H(c_l \mid C_{(<l)})$ through hierarchical vector quantization, where each level refines the partition of the representation space. Conditioning on deeper codes progressively shrinks the conditional variance of $\mathbf{z}$, which also contributes to lowering the global reconstruction uncertainty $H(\mathbf{z} \mid C)$.

Second, GAOQ explicitly targets prefix-dependent ambiguity $I(\mathbf{z};C_{(<l)} \mid c_l)$ by enforcing global alignment of code indices. At each level, child cluster centroids are first centered by subtracting their parent centroid, aligning cross-prefix representations to a common origin. We then introduce a set of approximately orthogonal reference directions shared across all parent nodes. Code indices are assigned by matching these centered vectors to the reference directions using Hungarian Matching based on cosine similarity. This procedure ensures that the same index within each level corresponds to a consistent semantic direction across different prefixes. This reduces both $I(\mathbf{z};C_{(<l)} \mid c_l)$ and $H(\mathbf{z} \mid c_l,C_{(<l)})$ and leads to lower $\sum_l H(\mathbf{z} \mid c_l)$.

Together, GAOQ minimizes all three terms in Eqn.~\eqref{equation:obj} while enforcing the entropy constraint on $H(c_l)$ via balanced K-Means. As a result, it produces compact, semantically stable, and autoregression-friendly SIDs that preserve task-relevant information for downstream generative recommendation. See Algorithm~\ref{alg1} in the appendix for more details.

\section{Experimental Settings}
We conduct extensive experiments to evaluate the effectiveness, interpretability, and efficiency of ReSID. 

Specifically, we seek to answer the following \textbf{Research Questions}:
\textbf{(i)} Does ReSID outperform strong sequential and SID-based baselines when side information is accessible?
\textbf{(ii)} How do FAMAE and GAOQ individually contribute to the overall performance?
\textbf{(iii)} Do the proposed task-aware embedding metrics reliably predict downstream SID performance?
\textbf{(iv)} Does ReSID improve the efficiency of SID tokenization relative to existing SID pipelines?

\textbf{Datasets.}
We evaluate ReSID on \textbf{ten} subsets of the Amazon-2023 review dataset~\cite{amazon2023}, including \textit{Musical Instruments}, \textit{Video Games}, \textit{Industrial \& Scientific}, \textit{Baby Products}, \textit{Arts, Crafts \& Sewing}, \textit{Sports \& Outdoors}, \textit{Toys \& Games}, \textit{Health \& Household}, \textit{Beauty \& Personal Care}, and \textit{Books}. We follow standard practice~\cite{letter,eager,etegrec} to preprocess the datasets. See detailed preprocessing steps and dataset statistics in Appendix~\ref{sec:appendix_dataset}.

\textbf{Fair comparison between SID- and item-ID-based recommendation.} A critical but often overlooked confounder in prior SID evaluations is that \emph{SID pipelines typically exploit rich item metadata, while sequential baselines are trained with item-IDs only}—making reported gains difficult to attribute to the modeling paradigm itself. Based on this, we choose our baselines as follows.

\textbf{Compared Methods.}
We compare ReSID with three categories of baselines:
\emph{(i) Sequential recommenders (item-ID-only):} HGN~\cite{hgn}, SASRec~\cite{sasrec}, BERT4Rec~\cite{bert4rec}, and \SthreeRec~\cite{s3rec}.
\emph{(ii) Sequential recommenders with structured features:} the corresponding $^*$ variants (e.g., HGN$^*$) of the above models augmented with side-info fields.
\emph{(iii) SID-based generative recommenders:} TIGER~\cite{tiger}, LETTER~\cite{letter}, EAGER~\cite{eager}, UNGER~\cite{unger}, and ETEGRec~\cite{etegrec}, which tokenize items into SIDs and perform autoregressive generation.
Detailed descriptions are provided in Appendix~\ref{sec:appendix_baseline}. Implementation details and hyperparameter settings are provided in Appendix~\ref{sec:appendix_implement}.

\textbf{Evaluation Settings.}
Following prior work~\cite{tiger,letter,etegrec}, we evaluate recommendation performance using Recall@K and NDCG@K with $K \in \{5, 10\}$. We adopt the leave-one-out protocol in which each user's last interaction is used for testing and the second-to-last for validation.

\section{Experimental Results}

\subsection{Overall Results under Fair Comparison} \label{sec:overall_res}
Table~\ref{tab:exp_main_results_summary} reports overall average relative improvements of ReSID over baselines on ten Amazon-2023 subsets (full results in Appendix~\ref{sec:appendix_main_results}). We have the following observations.

\textbf{5.1.1 Overall performance.}
ReSID consistently achieves the best performance across all metrics and datasets. To the best of our knowledge, ReSID is the first SID-based method that consistently outperforms strong item-ID–based sequential recommenders even when they are augmented with structured side information.

\textbf{5.1.2 Impact of structured features on sequential baselines.}
Augmenting sequential recommenders with structured feature fields yields substantial gains: models such as SASRec$^{*}$ and BERT4Rec$^{*}$ significantly outperform their ID-only variants and often match or exceed prior SID-based pipelines. This indicates that many previously reported SID improvements arise from \emph{additional side information} rather than the tokenization paradigm itself, highlighting the importance of matching side-info augmentations for fairness.

\textbf{5.1.3 ReSID vs. prior SID-based methods.}
Compared with the best SID baseline, LETTER, ReSID achieves average improvements of \textbf{16.0\%}/\textbf{13.8\%} on Recall@5/10 and \textbf{16.2\%}/\textbf{14.9\%} on NDCG@5/10. In contrast, prior SID methods are on par with and often inferior to SASRec$^{*}$. These show that effective SID tokenization needs objective alignment throughout E-, Q-, and G-stages, and validate ReSID’s successful shift beyond semantic-centric paradigms.

\begin{table}[t]
    \caption{
    Main results (ReSID's average \textbf{relative improvement} over baselines, \emph{averaged over ten datasets for each metric}). The smallest improvement is in bold, indicating the strongest baseline. See Appendix Table~\ref{tab:appx_main_results} for full results.
    }
    \label{tab:exp_main_results_summary}

    \centering
    \small
    \setlength{\tabcolsep}{8pt}
    \renewcommand{\arraystretch}{1.15}

    \begin{tabular}{@{}lcccc@{}}
    \toprule
    Model          & R@5     & R@10    & N@5     & N@10    \\
    \midrule
    HGN            & 79.91\% & 62.16\% & 97.21\% & 80.51\% \\
    SASRec         & 47.84\% & 37.52\% & 56.14\% & 47.56\% \\
    BERT4Rec       & 52.63\% & 41.73\% & 60.08\% & 51.43\% \\
    \SthreeRec     & 39.56\% & 29.13\% & 47.18\% & 38.58\% \\
    \midrule
    HGN$^*$        & 81.97\% & 67.86\% & 90.41\% & 79.29\% \\
    SASRec$^*$     & \textbf{13.05\%} & \textbf{3.75\%}  & 29.24\% & 18.38\% \\
    BERT4Rec$^*$   & 20.63\% & 12.32\% & 27.45\% & 20.41\% \\
    \SthreeRec$^*$ & 20.76\% & 12.79\% & 28.93\% & 21.57\% \\
    \midrule
    TIGER          & 22.21\% & 19.47\% & 23.66\% & 21.50\% \\
    LETTER         & 16.03\% & 13.81\% & \textbf{16.17\%} & \textbf{14.86\%} \\
    EAGER          & 37.30\% & 31.17\% & 39.02\% & 35.24\% \\
    UNGER          & 30.44\% & 26.51\% & 31.82\% & 29.18\% \\
    ETEGRec        & 47.07\% & 40.30\% & 50.94\% & 45.86\% \\
    \bottomrule
    \end{tabular}
\end{table}

\textbf{5.1.4 Effect of collaborative signal injection.}
Methods that incorporate collaborative signals during SID tokenization (LETTER, ReSID) consistently outperform purely semantic tokenizers such as TIGER. ReSID further improves over LETTER by preserving collaborative information \emph{throughout} both representation learning and quantization, rather than balancing it with semantic objectives as in LETTER, EAGER, and UNGER. These results suggest that maintaining collaborative structures is more critical for reducing autoregressive decoding uncertainty, while ReSID retains only the minimal necessary semantics through structured, feature-based representation learning with FAMAE.

\textbf{5.1.5 ``End-to-end'' SID learning is suboptimal.}
ETEGRec, which jointly optimizes SID tokenization and downstream recommendation loss, underperforms ReSID significantly despite its tighter coupling. This supports our analysis (Sec.~\ref{sec:rethinking_sid_pipeline}) that end-to-end SID supervision is inherently unstable: since SIDs serve as both intermediate representations and training targets, directly backpropagating task loss through the quantization stage can distort the code space. ReSID avoids this pitfall by decoupling representation learning, quantization, and recommender stages, yielding more stable and predictive tokenization.

\begin{table}[t]
    \caption{
    Ablation study (ReSID's average relative improvement over its variants, \emph{averaged over three datasets for each metric}). See Appendix Table~\ref{tab:appx_ablation} for full results.
    }
    \label{tab:exp_ablation_summary}

    \centering
    \small
    \setlength{\tabcolsep}{8pt}
    \renewcommand{\arraystretch}{1.15}

    \begin{tabular}{@{}lcccc@{}}
    \toprule
    Variants & R@5     & R@10   & N@5     & N@10   \\ \midrule
    E1       & 5.40\%  & 4.60\% & 4.12\%  & 3.86\% \\
    E2       & 11.05\% & 8.45\% & 10.73\% & 9.22\% \\
    E3       & 12.38\% & 8.76\% & 11.68\% & 9.61\% \\ \midrule
    Q1       & 5.64\%  & 4.18\% & 4.88\%  & 4.16\% \\
    Q2       & 5.15\%  & 1.41\% & 4.72\%  & 2.49\% \\ \bottomrule
    \end{tabular}
\end{table}

\subsection{Ablation Study}
To disentangle the contributions of the E-stage (FAMAE) and the Q-stage (GAOQ), we conduct ablations on three Amazon-2023 datasets and report results in Table~\ref{tab:exp_ablation_summary}.

\textbf{Variants.}
We consider the following controlled replacements.
\emph{E-stage ablations} (GAOQ fixed):
(1) \textbf{E1} (LLM w/ GAOQ): replacing FAMAE with LLM embeddings;
(2) \textbf{E2} (SASRec w/ GAOQ): replacing FAMAE with SASRec representations (collaborative-centric);
(3) \textbf{E3} (BERT4Rec w/ GAOQ): replacing FAMAE with BERT4Rec representations;
\emph{Q-stage ablations} (FAMAE fixed):
(4) \textbf{Q1} (FAMAE w/ RQ-VAE): replacing GAOQ with RQ-VAE;
(5) \textbf{Q2} (FAMAE w/ Hierarchical K-Means).
We compare these variants with \textbf{ReSID}. The full results are shown in Table~\ref{tab:appx_ablation}.

\textbf{Observations.}
ReSID consistently outperforms all ablated variants. Compared with E1–E3, ReSID yields substantial gains, validating the \emph{predictive sufficiency} analysis of FAMAE: neither purely semantic embeddings nor purely collaborative representations are sufficient for downstream SID-based recommendation tasks. ReSID’s advantage over E2 and E3 further supports the importance of preserving structured feature identity, which is lost in standard sequence encoders, supporting our \emph{predictive superiority} analysis.

On the quantization side, ReSID outperforms Q1, confirming that GAOQ better reduces autoregressive decoding uncertainty than RQ-VAE, and that minimizing reconstruction error alone is insufficient for downstream tasks. ReSID outperforms Q2, demonstrating the necessity of global index alignment to reduce the reconstruction error and index ambiguity beyond the locally indexed hierarchical K-Means.

Overall, the strongest results are obtained when recommendation-native representations and globally aligned quantization are combined to preserve task-relevant information while improving sequential predictability.

\subsection{Item Embedding Quality Metrics}
As introduced in Sec.~\ref{sec:metrics}, we evaluate whether the proposed task-aware proxy metrics reflect downstream SID performance. We conduct experiments on two datasets, \textit{Musical Instruments} and \textit{Baby Products}, tracking Metric~1 and Metric~2 at multiple checkpoints during FAMAE training and examining their relationship with downstream results under a fixed Q/G-stage setup.

Specifically, for each dataset, we select several intermediate checkpoints, freeze the learned item embeddings, construct SIDs using the same GAOQ quantizer, and train an identical SID generator. As shown in Fig.~\ref{fig:exp_metric_corr}, downstream R@10 consistently increases as both metrics improve across both datasets, indicating a strong positive association between embedding quality and SID performance. These results demonstrate that the proposed metrics serve as lightweight, task-aware diagnostics for E-stage FAMAE representations without requiring repeated end-to-end retraining.

\begin{figure}[t]
    \centering
    \includegraphics[width=0.48\linewidth]{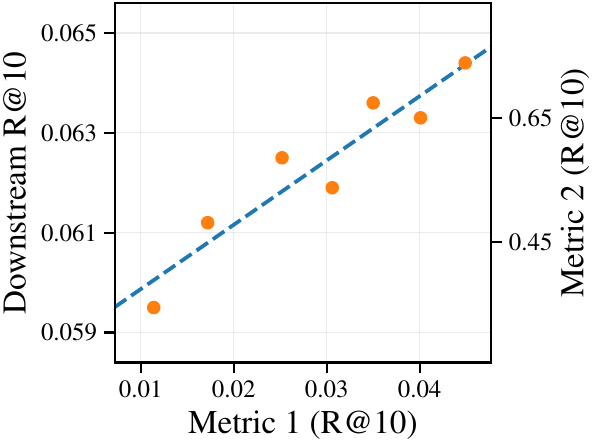}
    \hfill
    \includegraphics[width=0.48\linewidth]{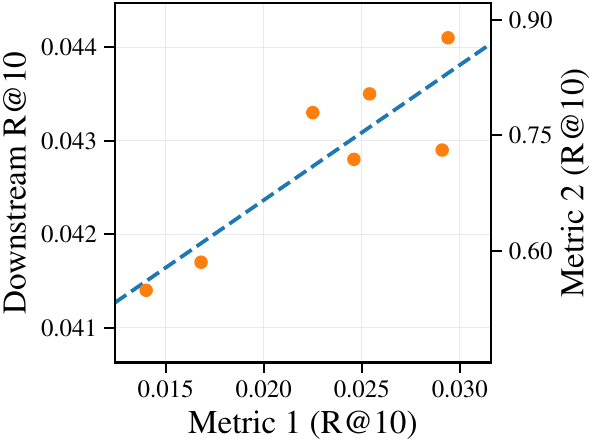}
    \caption{
    Downstream R@10 at selected FAMAE training steps plotted against Metric~1 (R@10 of the target item when all fields are masked). The right y-axis shows the corresponding Metric~2 (R@10 of the target item-ID when only the item-ID field is masked). Left: \textit{Musical Instruments}. Right: \textit{Baby Products}.
    }
    \label{fig:exp_metric_corr}
\end{figure}

\begin{table}[t]
    \caption{
    Running time comparison of the quantization stage for different methods on various datasets (in minutes).
    } 
    \label{tab:exp_wall_clock}

    \centering
    \small

    \begin{tabular}{lccc}
    \toprule
    \textbf{Dataset} & \textbf{LETTER} & \textbf{TIGER} & \textbf{ReSID} \\
    \midrule
    Arts Crafts \& Sewing    & 3356.64 & 224.29 & 27.38 \\
    Health \& Household      & 6537.30 & 371.71 & 71.89 \\
    Beauty \& Personal Care  & 7379.83 & 423.89 & 95.96 \\
    \bottomrule
    \end{tabular}
\end{table}

\subsection{Efficiency of SID Tokenization}
We evaluate SID tokenization efficiency by measuring the wall-clock runtime of the \emph{quantization stage} under the same hardware setting, reported in Table~\ref{tab:exp_wall_clock}. We compare ReSID with TIGER and LETTER as representative \emph{training-based} SID methods: TIGER relies on purely semantic tokenization, while LETTER injects collaborative signals through RQ-VAE training and represents the strongest prior baseline.

ReSID is the most efficient tokenizer among all compared methods. LETTER is the slowest, requiring \textbf{77$\times$--122$\times$} more time than ReSID due to its expensive optimization-based tokenizer training. TIGER is also slower, incurring approximately a \textbf{5$\times$} overhead relative to ReSID. These results demonstrate that ReSID improves not only recommendation effectiveness but also the computational efficiency and practicality of SID tokenization at scale.

We omit the runtime of the representation learning. FAMAE can be trained efficiently, with a cost comparable to light sequential models like SASRec. In contrast, prior SID pipelines rely on heterogeneous encoders (e.g., pretrained sequential models and large foundation models) whose costs are either amortized or unspecified, making a fair and meaningful time comparison difficult.

\subsection{Additional Experimental Results}
Due to space constraints, we report additional experimental analyses in the appendix. Specifically:
(1) Section~\ref{sec:sensitivity_branching} studies ReSID's sensitivity to key hyperparameters;
(2) Section~\ref{sec:scaling_law_result} shows that SIDs produced by ReSID exhibit more favorable scaling behavior than semantic-centric tokenizers;
(3) Section~\ref{sec:embedding_visual} visualizes item embeddings learned by FAMAE, demonstrating that FAMAE uniquely preserves both task-relevant semantic structures and collaborative patterns;
(4) Section~\ref{sec:task_align_analysis_famae} validates how FAMAE’s design aligns with downstream sequential decoding objectives;
(5) Section~\ref{sec:GAOQ_lower_ambiguity} provides direct evidence that GAOQ reduces reconstruction ambiguity through global index alignment.

\section{Limitations and Future Work}
While FAMAE provides task-aware metrics for embedding quality, principled diagnostics for GAOQ remain an open challenge. Moreover, although ReSID improves SID construction, SID-based generative models converge substantially more slowly (tens of times) than item-ID-based methods such as SASRec. We leave them for future work.

\section{Conclusion}
We introduce \textbf{ReSID}, a principled recommendation-native SID framework that aligns representation learning and quantization objectives with the information requirements of generative recommendation. Through FAMAE and GAOQ, ReSID produces compact and predictable SIDs efficiently without foundation models. Both theoretical analysis and extensive empirical results demonstrate the superiority of ReSID, which, to the best of our knowledge, is the first SID-based approach to outperform strong item-ID baselines augmented with side information.

\section*{Acknowledgements}
This paper is partially supported by National Natural Science Foundation of China (NO. U23A20313, 62372471) and The Science Foundation for Distinguished Young Scholars of Hunan Province (NO. 2023JJ10080). We are grateful for resources from the High Performance Computing Center of Central South University.

% \textbf{Do not} include acknowledgements in the initial version of the paper

% If a paper is accepted, the final camera-ready version can (and usually should)
% include acknowledgements.  Such acknowledgements should be placed at the end of
% the section, in an unnumbered section that does not count towards the paper
% page limit. Typically, this will include thanks to reviewers who gave useful
% comments, to colleagues who contributed to the ideas, and to funding agencies
% and corporate sponsors that provided financial support.

\section*{Impact Statement}
This paper advances machine learning for recommender systems, specifically in representation learning and discrete tokenization for recommender systems. 
The techniques proposed are methodological in nature and are intended to improve the performance and scalability of recommendation models. 
We do not foresee any significant negative societal or ethical consequences arising directly from this work.

% Authors are \textbf{required} to include a statement of the potential broader
% impact of their work, including its ethical aspects and future societal
% consequences. This statement should be in an unnumbered section at the end of
% the paper (co-located with Acknowledgements -- the two may appear in either
% order, but both must be before References), and does not count toward the paper
% page limit. In many cases, where the ethical impacts and expected societal
% implications are those that are well established when advancing the field of
% Machine Learning, substantial discussion is not required, and a simple
% statement such as the following will suffice:

% ``This paper presents work whose goal is to advance the field of Machine
% Learning. There are many potential societal consequences of our work, none
% which we feel must be specifically highlighted here.''

% The above statement can be used verbatim in such cases, but we encourage
% authors to think about whether there is content which does warrant further
% discussion, as this statement will be apparent if the paper is later flagged
% for ethics review.

\bibliographystyle{icml2026}
\bibliography{References}

\newpage
\appendix
\onecolumn

% main results
% 10 datasets, 4 ID-based baseline, 4 (ID + feature)-based baseline and 5 SID-based baseline

\begin{landscape}
\begin{center}    
    \captionof{table}{
    Full main results. Datasets: MI = Musical Instruments, VG = Video Games, IS = Industrial \& Scientific, BP = Baby Products, ACS = Arts, Crafts \& Sewing, SO = Sports \& Outdoors, TG = Toys \& Games, HH = Health \& Household, BPC = Beauty \& Personal Care, BK = Books. For each dataset, we report Recall (R@5, R@10) and NDCG (N@5, N@10). The best result in each column is shown in \textbf{bold}, and the second-best is \underline{underlined}.
    }
    \label{tab:appx_main_results}

    \footnotesize
    \setlength{\tabcolsep}{2.8pt}
    \renewcommand{\arraystretch}{1.05}

    % ===================== (A) 5 datasets =====================
    \textbf{(A) MI / VG / IS / BP / ACS}
    \vspace{2pt}

    \begin{tabular}{@{}l*{20}{c}@{}}
    \toprule
    Model
    & \multicolumn{4}{c}{MI}
    & \multicolumn{4}{c}{VG}
    & \multicolumn{4}{c}{IS}
    & \multicolumn{4}{c}{BP}
    & \multicolumn{4}{c}{ACS} \\
    \cmidrule(lr){2-5}\cmidrule(lr){6-9}\cmidrule(lr){10-13}\cmidrule(lr){14-17}\cmidrule(lr){18-21}
    & R@5 & R@10 & N@5 & N@10
    & R@5 & R@10 & N@5 & N@10
    & R@5 & R@10 & N@5 & N@10
    & R@5 & R@10 & N@5 & N@10
    & R@5 & R@10 & N@5 & N@10 \\
    \midrule
    HGN            & 0.0309 & 0.0498 & 0.0192 & 0.0253 & 0.0408 & 0.0685 & 0.0248 & 0.0337 & 0.0214 & 0.0361 & 0.0132 & 0.0179 & 0.0164 & 0.0276 & 0.0101 & 0.0137 & 0.0190 & 0.0314 & 0.0119 & 0.0159 \\
    SASRec         & 0.0320 & 0.0528 & 0.0199 & 0.0265 & 0.0509 & 0.0820 & 0.0322 & 0.0422 & 0.0216 & 0.0349 & 0.0133 & 0.0176 & 0.0206 & 0.0338 & 0.0131 & 0.0174 & 0.0218 & 0.0353 & 0.0138 & 0.0182 \\
    BERT4Rec       & 0.0324 & 0.0517 & 0.0206 & 0.0268 & 0.0479 & 0.0783 & 0.0303 & 0.0401 & 0.0207 & 0.0341 & 0.0131 & 0.0174 & 0.0204 & 0.0338 & 0.0129 & 0.0172 & 0.0211 & 0.0343 & 0.0134 & 0.0176 \\
    \SthreeRec     & 0.0336 & 0.0543 & 0.0212 & 0.0278 & 0.0495 & 0.0801 & 0.0315 & 0.0413 & 0.0210 & 0.0349 & 0.0135 & 0.0179 & 0.0213 & 0.0360 & 0.0135 & 0.0182 & 0.0228 & 0.0370 & 0.0145 & 0.0191 \\
    \midrule
    HGN$^*$        & 0.0287 & 0.0469 & 0.0183 & 0.0242 & 0.0331 & 0.0541 & 0.0212 & 0.0280 & 0.0192 & 0.0321 & 0.0122 & 0.0163 & 0.0169 & 0.0278 & 0.0108 & 0.0143 & 0.0166 & 0.0271 & 0.0107 & 0.0141 \\
    SASRec$^*$     & 0.0397 & \underline{0.0639} & 0.0248 & 0.0326 & 0.0532 & 0.0914 & 0.0295 & 0.0417 & \underline{0.0333} & \textbf{0.0538} & 0.0185 & 0.0250 & 0.0257 & \underline{0.0425} & 0.0157 & 0.0211 & 0.0275 & 0.0464 & 0.0157 & 0.0218 \\
    BERT4Rec$^*$   & 0.0377 & 0.0614 & 0.0241 & 0.0317 & 0.0537 & 0.0879 & 0.0337 & 0.0447 & 0.0296 & 0.0488 & 0.0186 & 0.0248 & 0.0241 & 0.0390 & 0.0155 & 0.0202 & 0.0275 & 0.0439 & 0.0176 & 0.0229 \\
    \SthreeRec$^*$ & 0.0380 & 0.0611 & 0.0242 & 0.0317 & 0.0516 & 0.0843 & 0.0326 & 0.0431 & 0.0274 & 0.0441 & 0.0169 & 0.0223 & 0.0237 & 0.0386 & 0.0152 & 0.0200 & 0.0261 & 0.0420 & 0.0168 & 0.0219 \\
    \midrule
    TIGER          & 0.0385 & 0.0592 & 0.0251 & 0.0318 & 0.0569 & 0.0884 & 0.0374 & 0.0475 & 0.0297 & 0.0460 & 0.0192 & 0.0244 & 0.0256 & 0.0406 & 0.0169 & 0.0217 & 0.0285 & 0.0446 & 0.0187 & 0.0239 \\
    LETTER         & \underline{0.0406} & 0.0623 & \underline{0.0269} & \underline{0.0338} & \underline{0.0592} & \underline{0.0914} & \underline{0.0389} & \underline{0.0493} & 0.0307 & 0.0482 & \underline{0.0197} & \underline{0.0254} & \underline{0.0261} & 0.0416 & \underline{0.0171} & \underline{0.0221} & \underline{0.0297} & \underline{0.0464} & \underline{0.0196} & \underline{0.0249} \\
    EAGER          & 0.0327 & 0.0523 & 0.0210 & 0.0273 & 0.0547 & 0.0863 & 0.0357 & 0.0458 & 0.0252 & 0.0406 & 0.0163 & 0.0212 & 0.0199 & 0.0331 & 0.0126 & 0.0168 & 0.0240 & 0.0375 & 0.0155 & 0.0198 \\
    UNGER          & 0.0362 & 0.0567 & 0.0233 & 0.0299 & 0.0547 & 0.0859 & 0.0354 & 0.0454 & 0.0232 & 0.0374 & 0.0150 & 0.0196 & 0.0221 & 0.0355 & 0.0144 & 0.0187 & 0.0253 & 0.0395 & 0.0164 & 0.0210 \\
    ETEGRec        & 0.0372 & 0.0579 & 0.0244 & 0.0310 & 0.0558 & 0.0869 & 0.0367 & 0.0466 & 0.0251 & 0.0393 & 0.0166 & 0.0212 & 0.0229 & 0.0369 & 0.0151 & 0.0196 & 0.0249 & 0.0389 & 0.0162 & 0.0206 \\
    \midrule
    ReSID          & \textbf{0.0417} & \textbf{0.0645} & \textbf{0.0273} & \textbf{0.0346} & \textbf{0.0597} & \textbf{0.0927} & \textbf{0.0396} & \textbf{0.0501} & \textbf{0.0340} & \underline{0.0512} & \textbf{0.0218} & \textbf{0.0273} & \textbf{0.0285} & \textbf{0.0441} & \textbf{0.0186} & \textbf{0.0236} & \textbf{0.0313} & \textbf{0.0485} & \textbf{0.0205} & \textbf{0.0261} \\
    \bottomrule
    \end{tabular}

    \vspace{10pt}

    % ===================== (B) 5 datasets =====================
    \textbf{(B) SO / TG / HH / BPC / BK}
    \vspace{2pt}

    \begin{tabular}{@{}l*{20}{c}@{}}
    \toprule
    Model
    & \multicolumn{4}{c}{SO}
    & \multicolumn{4}{c}{TG}
    & \multicolumn{4}{c}{HH}
    & \multicolumn{4}{c}{BPC}
    & \multicolumn{4}{c}{BK} \\
    \cmidrule(lr){2-5}\cmidrule(lr){6-9}\cmidrule(lr){10-13}\cmidrule(lr){14-17}\cmidrule(lr){18-21}
    & R@5 & R@10 & N@5 & N@10
    & R@5 & R@10 & N@5 & N@10
    & R@5 & R@10 & N@5 & N@10
    & R@5 & R@10 & N@5 & N@10
    & R@5 & R@10 & N@5 & N@10 \\
    \midrule
    HGN            & 0.0112 & 0.0184 & 0.0070 & 0.0093 & 0.0129 & 0.0214 & 0.0077 & 0.0105 & 0.0112 & 0.0189 & 0.0068 & 0.0093 & 0.0123 & 0.0207 & 0.0076 & 0.0103 & 0.0203 & 0.0367 & 0.0113 & 0.0166 \\
    SASRec         & 0.0140 & 0.0231 & 0.0089 & 0.0118 & 0.0169 & 0.0270 & 0.0105 & 0.0138 & 0.0137 & 0.0227 & 0.0087 & 0.0115 & 0.0157 & 0.0254 & 0.0101 & 0.0132 & 0.0281 & 0.0459 & 0.0170 & 0.0227 \\
    BERT4Rec       & 0.0131 & 0.0216 & 0.0082 & 0.0110 & 0.0159 & 0.0255 & 0.0101 & 0.0131 & 0.0135 & 0.0224 & 0.0086 & 0.0114 & 0.0150 & 0.0246 & 0.0095 & 0.0126 & 0.0280 & 0.0456 & 0.0173 & 0.0230 \\
    \SthreeRec     & 0.0153 & 0.0249 & 0.0097 & 0.0128 & 0.0186 & 0.0297 & 0.0117 & 0.0153 & 0.0150 & 0.0245 & 0.0094 & 0.0125 & 0.0171 & 0.0282 & 0.0110 & 0.0145 & 0.0318 & 0.0530 & 0.0182 & 0.0251 \\
    \midrule
    HGN$^*$        & 0.0121 & 0.0197 & 0.0078 & 0.0102 & 0.0119 & 0.0193 & 0.0075 & 0.0098 & 0.0124 & 0.0202 & 0.0078 & 0.0103 & 0.0137 & 0.0219 & 0.0088 & 0.0115 & 0.0231 & 0.0397 & 0.0138 & 0.0191 \\
    SASRec$^*$     & \underline{0.0197} & \underline{0.0319} & 0.0122 & \underline{0.0161} & \underline{0.0243} & \textbf{0.0401} & 0.0127 & 0.0178 & 0.0176 & \underline{0.0287} & 0.0108 & 0.0144 & \underline{0.0208} & \underline{0.0340} & 0.0126 & 0.0169 & 0.0369 & 0.0629 & 0.0202 & 0.0285 \\
    BERT4Rec$^*$   & 0.0175 & 0.0286 & 0.0111 & 0.0147 & 0.0223 & 0.0356 & 0.0138 & 0.0181 & 0.0168 & 0.0274 & 0.0106 & 0.0140 & 0.0192 & 0.0312 & 0.0123 & 0.0161 & 0.0337 & 0.0564 & 0.0199 & 0.0272 \\
    \SthreeRec$^*$ & 0.0184 & 0.0297 & 0.0117 & 0.0154 & 0.0209 & 0.0338 & 0.0128 & 0.0170 & 0.0170 & 0.0274 & 0.0107 & 0.0140 & 0.0198 & 0.0322 & 0.0125 & 0.0165 & \underline{0.0381} & \underline{0.0639} & 0.0207 & \underline{0.0290} \\
    \midrule
    TIGER          & 0.0180 & 0.0279 & 0.0119 & 0.0151 & 0.0214 & 0.0322 & 0.0141 & 0.0176 & 0.0180 & 0.0276 & 0.0121 & 0.0152 & 0.0192 & 0.0299 & 0.0125 & 0.0160 & 0.0265 & 0.0408 & 0.0171 & 0.0219 \\
    LETTER         & 0.0187 & 0.0288 & \underline{0.0124} & 0.0156 & 0.0221 & 0.0340 & \underline{0.0147} & \underline{0.0185} & \underline{0.0180} & 0.0277 & \underline{0.0121} & \underline{0.0153} & 0.0203 & 0.0315 & \underline{0.0133} & \underline{0.0169} & 0.0308 & 0.0457 & 0.0209 & 0.0256 \\
    EAGER          & 0.0145 & 0.0236 & 0.0095 & 0.0125 & 0.0191 & 0.0289 & 0.0127 & 0.0159 & 0.0138 & 0.0223 & 0.0089 & 0.0116 & 0.0164 & 0.0265 & 0.0107 & 0.0139 & 0.0342 & 0.0497 & \underline{0.0234} & 0.0284 \\
    UNGER          & 0.0164 & 0.0260 & 0.0109 & 0.0140 & 0.0200 & 0.0302 & 0.0132 & 0.0165 & 0.0168 & 0.0258 & 0.0112 & 0.0141 & 0.0182 & 0.0286 & 0.0119 & 0.0152 & 0.0309 & 0.0468 & 0.0203 & 0.0254 \\
    ETEGRec        & 0.0148 & 0.0236 & 0.0097 & 0.0125 & 0.0178 & 0.0270 & 0.0116 & 0.0146 & 0.0143 & 0.0229 & 0.0092 & 0.0120 & 0.0177 & 0.0274 & 0.0115 & 0.0146 & 0.0176 & 0.0284 & 0.0110 & 0.0145 \\
    \midrule
    ReSID          & \textbf{0.0210} & \textbf{0.0323} & \textbf{0.0137} & \textbf{0.0174} & \textbf{0.0261} & \underline{0.0392} & \textbf{0.0174} & \textbf{0.0216} & \textbf{0.0208} & \textbf{0.0317} & \textbf{0.0140} & \textbf{0.0175} & \textbf{0.0230} & \textbf{0.0357} & \textbf{0.0150} & \textbf{0.0191} & \textbf{0.0530} & \textbf{0.0737} & \textbf{0.0369} & \textbf{0.0436} \\
    \bottomrule
    \end{tabular}

\end{center}
\end{landscape}

% ablation
% representation
%   E1: text + GARQ; E2: sasrec (ID) + GARQ; E3: bert4rec (ID) + GARQ
% quantization
%   Q1: FAMAE + rqvae; Q2: FAMAE + GARQ (w/o global alignment)

\begin{table}[t]
    \caption{
    Full Ablation study. Datasets: MI = Musical Instruments, IS = Industrial \& Scientific, BP = Baby Products. Each dataset reports Recall (R@5, R@10) and NDCG (N@5, N@10). The best result in each column is shown in \textbf{bold}.
    }
    \label{tab:appx_ablation}

    \centering
    \small
    \setlength{\tabcolsep}{6.5pt}
    \renewcommand{\arraystretch}{1.15}

    \begin{tabular}{lcccccccccccc}
    \toprule
    Model & \multicolumn{4}{c}{MI} & \multicolumn{4}{c}{IS} & \multicolumn{4}{c}{BP} \\
    \cmidrule(lr){2-5}\cmidrule(lr){6-9}\cmidrule(lr){10-13}
    & R@5 & R@10 & N@5 & N@10 & R@5 & R@10 & N@5 & N@10 & R@5 & R@10 & N@5 & N@10 \\
    \midrule
    E1    & 0.0413 & 0.0629 & 0.0272 & 0.0341 & 0.0303 & 0.0472 & 0.0198 & 0.0252 & 0.0276 & 0.0428 & 0.0183 & 0.0232 \\
    E2    & 0.0400 & 0.0621 & 0.0261 & 0.0332 & 0.0320 & 0.0496 & 0.0208 & 0.0264 & 0.0232 & 0.0373 & 0.0152 & 0.0197 \\
    E3    & 0.0400 & 0.0624 & 0.0261 & 0.0333 & 0.0309 & 0.0488 & 0.0201 & 0.0258 & 0.0232 & 0.0373 & 0.0153 & 0.0198 \\
    \midrule
    Q1    & 0.0396 & 0.0609 & 0.0260 & 0.0328 & 0.0320 & 0.0497 & 0.0207 & 0.0264 & 0.0271 & 0.0424 & 0.0179 & 0.0228 \\
    Q2    & 0.0390 & 0.0630 & 0.0255 & 0.0332 & 0.0320 & 0.0505 & 0.0206 & 0.0266 & 0.0279 & 0.0438 & 0.0184 & 0.0235 \\
    \midrule
    ReSID & \textbf{0.0417} & \textbf{0.0645} & \textbf{0.0273} & \textbf{0.0346} & \textbf{0.0340} & \textbf{0.0512} & \textbf{0.0218} & \textbf{0.0273} & \textbf{0.0285} & \textbf{0.0441} & \textbf{0.0186} & \textbf{0.0236} \\
    \bottomrule
    \end{tabular}
\end{table}

\section{Proof for Proposition \ref{prop:sufficiency}}

\begin{proof}
    For each field $k$ we denote the true conditional distribution as $p_k(f_T^{(k)} \mid \mathbf{h}_T)$, and the approximated distribution as $q_{\theta,k}(f_T^{(k)} \mid \mathbf{h}_T)$. Given the negative log-likelihood for field $k$ as:
    \begin{equation*}
        \mathcal{L}_k(\theta):=\mathbb{E}\big[-\log q_{\theta,k}(f_T^{(k)} \mid \mathbf{h}_T) \big],
    \end{equation*}
    Using the standard cross-entropy decomposition, we have:
    \begin{equation*}
        \mathcal{L}_k(\theta)=H(f_T^{(k)} \mid \mathbf{h}_T)+\mathbb{E}\big[D_{\mathrm{KL}}(p_k(f_T^{(k)} \mid \mathbf{h}_T)||q_{\theta,k}(f_T^{(k)} \mid \mathbf{h}_T))\big].
    \end{equation*}
    Since the KL divergence is non-negative, we have:
    \begin{equation*}
        -\mathcal{L}_k(\theta) \le -H(f_T^{(k)} \mid \mathbf{h}_T)
    \end{equation*}
    According to mutual information identity, we have:
    \begin{equation*}
        I(\mathbf{h}_T;f_T^{(k)})=H(f_T^{(k)}) - H(f_T^{(k)} \mid \mathbf{h}_T) \ge H(f_T^{(k)}) - \mathcal{L}_k(\theta).
    \end{equation*}
    Multiply both sides with the predefined mask-weighted coefficient $w_k(\ge0)$ and sum over $k$, we have:
    \begin{align}
        \sum_{k=1}^{J}w_k I(\mathbf{h}_T;f_T^{(k)})&\ge \sum_{k=1}^J w_k H(f_T^{(k)}) - \sum_{k=1}^Jw_k\mathcal{L}_k(\theta) \notag \\
        &= \sum_{k=1}^J w_k H(f_T^{(k)}) - \mathcal{L}_{\mathrm{FAMAE}}(\theta). \notag
    \end{align}
    Since $\sum_{k=1}^J w_k H(f_T^{(k)})$ is independent of $\theta$, minimizing $\mathcal{L}_{\mathrm{FAMAE}}(\theta)$ monotonically increases a variational lower bound on the mask-weighted sum $\sum_{k=1}^J w_k I(\mathbf{h}_T; f_T^{(k)})$.
\end{proof}

\section{Related Work}

\subsection{Sequential Recommendation.}
Sequential recommendation aims to model users' evolving preferences from their historical interaction sequences to predict future items.
Early neural approaches~\cite{narm,caser,gru4rec,srgnn,improvedrnns}, such as GRU4Rec~\cite{gru4rec}, adopt GRU-based RNNs to encode user behavior sequences. 
Recently, Transformer~\cite{transformer}-based models have been introduced into sequential recommendation. SASRec~\cite{sasrec} employs a unidirectional self-attention mechanism to model item–item dependencies across the entire sequence. Inspired by masked language modeling, BERT4Rec~\cite{bert4rec} formulates sequential recommendation as a masked item prediction task and leverages bidirectional self-attention to learn contextualized item representations. Building upon this framework, \SthreeRec~\cite{s3rec} further enhances representation learning by incorporating multiple self-supervised pre-training objectives.
Despite architectural differences, these methods generally follow a common framework: they embed each item into a high-dimensional embedding, summarize user behavior into a sequence representation, and predict the next item by scoring candidate items—typically via dot-product (or cosine) similarity between the sequence representation and item embeddings, combined with approximate nearest neighbor (ANN) search~\cite{faiss} for efficient retrieval.

\subsection{SID-based Generative Recommendation.}
Traditional sequential recommendation models assign each item a dedicated embedding and retrieve items via similarity matching~\cite{recformer}. However, this design necessitates massive embedding tables and requires an exhaustive comparison across the entire item space during retrieval, leading to substantial memory overhead and high computational costs~\cite{tiger,liger}. In contrast, SID-based Generative Recommendation represents items as discrete identifiers and formulates recommendation as a sequence generation problem~\cite{rpg}, directly generating the target item, avoiding explicit similarity search over the full item vocabulary.

TIGER~\cite{tiger} is a pioneering work in this direction. It follows a pipeline that first derives item representations from textual information, then encodes them into discrete SIDs using RQ-VAE, and finally employs an encoder–decoder model to generate the SIDs of the target item conditioned on historical interactions. During inference, candidate SIDs are generated via beam search and mapped back to items for Top-K recommendation.

Following this pipeline, a large body of subsequent work focuses on optimizing different components of the pipeline, with particular emphasis on improving the quality and expressiveness of SIDs. 
For example, CoST~\cite{cost} improves the training objective of RQ-VAE to preserve important neighborhood structures among items, while LIGER~\cite{liger} and COBRA~\cite{cobra} enable generative models to jointly model item SIDs and dense representations, providing richer item information and thereby improving prediction accuracy. 
Several methods further inject collaborative signals into SID learning. LETTER~\cite{letter} injects collaborative supervision directly into the RQ-VAE-based SID learning process; EAGER~\cite{eager} separately learns discrete identifiers for semantic content and user behavior; and UNGER~\cite{unger} learns item representations by jointly modeling semantic and collaborative information before deriving SIDs. In addition, ETEGRec~\cite{etegrec} aligns the learning of item identifiers with the recommendation objective in an end-to-end way, enabling SIDs to encode information that is more directly optimized for downstream recommendation.

Although prior work such as LETTER, EAGER, and UNGER explores different ways to incorporate collaborative signals into the original learning process, these methods primarily inject collaborative supervision by adding auxiliary objectives, which are not fully consistent with the original learning objective, leading to conflicting optimization signals. As a result, the learned identifiers are still not explicitly optimized to preserve the task-relevant information required by downstream generative recommendation. ETEGRec further pushes this direction by coupling identifier learning with downstream generative recommendation in an end-to-end manner. However, the continuously evolving identifiers make the generative model’s inputs and targets non-stationary, causing strong coupling and mutual interference between SID learning and sequence generation. ReSID differs by providing an information-theoretic framework that jointly aligns representation learning and SID quantization under a unified objective, resulting in a simpler and more stable learning pipeline and consistently strong empirical performance.

\section{Dataset} \label{sec:appendix_dataset}

% Density = interaction / (item*user)

\begin{table}[t]
    \caption{
    Statistics of the Datasets.
    }
    \label{tab:appx_dataset}

    \centering

    \begin{tabular}{lrrrr}
    \toprule
    Dataset & \#Users & \#Items & \#Interactions & Density \\
    \midrule
    Musical Instruments & 57,359 & 23,742 & 490,522 & 0.036\% \\
    Video Games & 94,515 & 24,685 & 772,218 & 0.033\% \\
    Industrial \& Scientific & 50,886 & 25,142 & 394,989 & 0.031\% \\
    Baby Products & 150,642 & 35,024 & 1,189,171 & 0.023\% \\
    Arts Crafts \& Sewing & 196,980 & 87,449 & 1,706,484 & 0.010\% \\
    Sports \& Outdoors & 409,309 & 151,411 & 3,333,753 & 0.005\% \\
    Toys \& Games & 431,411 & 156,537 & 3,652,250 & 0.005\% \\
    Health \& Household & 796,014 & 183,230 & 6,860,582 & 0.005\% \\
    Beauty \& Personal Care & 712,259 & 193,383 & 5,785,124 & 0.004\% \\
    Books & 775,503 & 485,218 & 8,680,494 & 0.002\% \\
    \bottomrule
    \end{tabular}
\end{table}

Following prior work~\cite{letter,eager,etegrec}, we adopt the standard 5-core setting, where users and items with fewer than five interactions are removed. The remaining interactions are chronologically ordered to construct user behavior sequences, and a leave-one-out strategy is employed for evaluation. For training, we employ a sliding window strategy with a maximum sequence length of 32. In cases where a target item in the evaluation set does not appear in the training data, it is added to the training set to ensure valid evaluation. In addition, we extract four types of structured side information for each item, including the store identifier and the first-, second-, and third-level category identifiers. Items lacking any of these features are filtered out. Table~\ref{tab:appx_dataset} reports detailed statistics of the resulting datasets.

\section{Compared Methods} \label{sec:appendix_baseline}
We compare ReSID with representative sequential recommendation models and recent generative recommendation methods.

\noindent\textbf{Sequential recommendation methods.}
\begin{itemize}
    \item \textbf{HGN}~\cite{hgn} models user preferences by jointly capturing short-term and long-term interests through a hierarchical gating mechanism over interaction sequences.
    \item \textbf{SASRec}~\cite{sasrec} employs a unidirectional Transformer with self-attention to model sequential dependencies and predicts the next item from historical interactions.
    \item \textbf{BERT4Rec}~\cite{bert4rec} introduces bidirectional Transformer encoding with masked item prediction, leveraging both left and right context to learn sequence representations.
    \item \textbf{\SthreeRec}~\cite{s3rec} enhances sequential recommendation via multiple self-supervised learning objectives based on mutual information maximization, improving representation learning for next-item prediction. For \SthreeRec, we follow the standard two-stage pipeline with self-supervised pretraining followed by SASRec-style fine-tuning on item-ID sequences.
\end{itemize}

\noindent\textbf{Generative recommendation methods.}
\begin{itemize}
    \item \textbf{TIGER}~\cite{tiger} utilizes a pretrained text encoder and RQ-VAE quantization to learn semantic identifiers for items, and performs generative recommendation by autoregressively decoding item identifiers.
    \item \textbf{LETTER}~\cite{letter} proposes a learnable tokenizer that extends RQ-VAE-based semantic identifiers by jointly incorporating hierarchical semantics, collaborative signals, and code assignment diversity for generative recommendation.
    \item \textbf{EAGER}~\cite{eager} integrates semantic and collaborative information via a two-stream generative architecture with shared encoding and separate decoding for enhanced collaborative modeling.
    \item \textbf{UNGER}~\cite{unger} integrates semantic and collaborative information into a unified item code by learning item representations via joint optimization of cross-modality alignment and next-item prediction in a sequential recommendation model.
    \item \textbf{ETEGRec}~\cite{etegrec} integrates item tokenization and generative recommendation into a unified end-to-end framework, jointly optimizing the tokenization and recommendation processes for improved performance.
\end{itemize}

\section{Implementation Details} \label{sec:appendix_implement}
Our experiments follow the standard three-stage pipeline: representation learning (E-stage), SID construction via quantization (Q-stage), and SID-based generative modeling (G-stage). We summarize the hyperparameter settings below.

\textbf{E-stage and sequential baselines.}
Unless otherwise specified, all sequential recommenders and the E-stage encoder in ReSID share the same model size and training configuration for controlled comparison.
We use an embedding/hidden size of 128, 2 Transformer layers, 4 attention heads, and an FFN dimension of 512 with ReLU activation, together with dropout rate of 0.1.
Feature embeddings have dimension 128 and are fused by sum-pooling when feature fields are used.
We optimize with AdamW (learning rate 0.001, weight decay $1.0\times10^{-5}$) using batch size 2048 and train for up to 500 epochs with early stopping patience 3, evaluating every epoch.
When a sampled classification objective is used (e.g., for large vocabularies), we sample 128 negatives and use scaled cosine similarity.
We use a cosine learning-rate scheduler.

\textbf{Q-stage.}
For ReSID, the dataset-specific branching factors $\{b_l\}$ are reported in Table~\ref{tab:appx_resid_cluster}.
For SID-based baselines (TIGER, LETTER, EAGER, UNGER, and ETEGRec), we follow the optimal quantization settings and other hyperparameters reported in their original papers and official implementations.
For methods requiring text-based item embeddings (TIGER, LETTER, EAGER, and UNGER), we follow prior work~\cite{tiger} and use the pretrained Sentence-T5-xxl~\cite{sentencet5} to obtain semantic item embeddings before discretization.

\begin{table}[t]
    \caption{
    ReSID's branching factors on each dataset. Datasets: MI = Musical Instruments, VG = Video Games, IS = Industrial \& Scientific, BP = Baby Products, ACS = Arts, Crafts \& Sewing, SO = Sports \& Outdoors, TG = Toys \& Games, HH = Health \& Household, BPC = Beauty \& Personal Care, BK = Books.
    }
    \label{tab:appx_resid_cluster}

    \centering
    \small
    \setlength{\tabcolsep}{4pt}

    \begin{tabular}{*{10}{c}}
    \toprule
    MI & VG & IS & BP & ACS & SO & TG & HH & BPC & BK \\
    \midrule
    (32,40,19) & (32,64,13) & (24,80,14) & (32,64,18) & (64,96,16) & (128,128,11) & (192,192,5) & (50,512,8) & (96,192,11) & (256,256,8) \\
    \bottomrule
    \end{tabular}
\end{table}

\textbf{G-stage.}
All SID-based generative recommenders (including ReSID and prior SID baselines) use the same T5-style encoder-decoder architecture for training on SID sequences.
We use 4 encoder layers and 4 decoder layers, hidden size 128, FFN dimension 512, and 4 attention heads per layer (key/value dimension 32), with dropout rate of 0.1.
We use AdamW (learning rate 0.005, weight decay $1.0\times10^{-5}$) with a cosine learning-rate scheduler, batch size 2048, and train for up to 500 epochs.
During inference, we use beam search with beam size 50 at each decoding step.

\section{Full Results of Main Experiments} \label{sec:appendix_main_results}

Table~\ref{tab:appx_main_results} reports the \emph{absolute} performance of all compared methods on each Amazon-2023 subset under the experimental settings described in Section~\ref{sec:overall_res}. These results complement the macro-averaged relative improvements reported in the main paper and provide detailed subset-wise comparisons for reproducibility and further analysis.

\section{Sensitivity to Branching Factors} \label{sec:sensitivity_branching}

\begin{table}[t]
    \caption{
    Sensitivity to branching factor at the first two SID levels on \textit{Musical Instruments}. We vary $(b_1,b_2)$ and report downstream Recall@10 (R@10). Each row fixes $b_1$ and sweeps $b_2$.
    }
    \label{tab:appx_sensitivity_branching}

    \centering
    \setlength{\tabcolsep}{4pt}
    \renewcommand{\arraystretch}{1.1}

    \begin{tabular}{@{}c*{8}{c}@{}}
    \toprule
    $b_1$ \textbackslash\ $b_2$ & 24 & 32 & 40 & 48 & 56 & 64 & 72 & 80 \\
    \midrule
    16 & 6.3068 & 6.3906 & 6.1548 & 6.2928 & 6.1897 & 6.3277 & 6.2928 & 6.2386 \\
    24 & 6.1985 & 6.3155 & 6.1950 & 6.3539 & 6.3609 & 6.2404 & 6.3260 & 6.1915 \\
    32 & 6.3469 & 6.3120 & 6.4465 & 6.3068 & 6.2509 & 6.2718 & 6.2945 & 6.2579 \\
    48 & 6.2806 & 6.2072 & 6.2229 & 6.1880 & 5.9941 & 6.0570 & 6.1897 & 6.2055 \\
    \bottomrule
    \end{tabular}
\end{table}

We investigate how \emph{branching factors} affect ReSID. In our hierarchical quantization, at each level $l$, \emph{each parent cluster} is partitioned into $b_l$ balanced child clusters (Alg.~\ref{alg1}), where $b_l$ is the branching factor (i.e., the number of children per parent) at that level. Varying $\{b_l\}$ controls the granularity and capacity of the discrete code space, trading off distortion under the discrete bottleneck (captured by $H(\mathbf{z} \mid C)$) and sequential uncertainty (captured by $\sum_l H(c_l \mid C_{(<l)})$), as discussed in Section~\ref{sec:method_gaoq}.

\textbf{Protocol.}
We conduct this study on \textit{Musical Instruments}. Given the item scale in our experiments (20K--400K items), we use a three-level SID throughout and vary the branching factors at the first two levels $(b_1,b_2)$ while keeping other settings fixed. The last level primarily serves to disambiguate items within each prefix $(c_1,c_2)$; its effective branching factor depends on the population of each prefix and can be auto-computed once the prefix levels are fixed, so we do not treat $b_3$ as a primary tuning knob.
Following the GAOQ design, global alignment is applied to non-root levels where indices would otherwise be locally assigned under different parent prefixes; the root-level codes are already globally defined by clustering in the original embedding space and thus do not require global alignment. We keep the default anchor setting ($g_l=b_l$ at aligned levels) and do not vary it separately.

\textbf{Results.}
Table~\ref{tab:appx_sensitivity_branching} shows a clear capacity-predictability trade-off consistent with Section~\ref{sec:method_gaoq}. When $(b_1,b_2)$ are too small, the code space is overly coarse and performance degrades due to insufficient capacity and higher quantization distortion. Increasing $(b_1,b_2)$ improves performance and reaches its best range at moderate branching factors (e.g., $(b_1,b_2)=(32,40)$ achieves the highest Recall@10 in our sweep), after which further increasing the branching factors yields diminishing returns and can slightly hurt performance as autoregressive decoding uncertainty increases. Overall, we observe the optimal performance when $b_1 \times b_2$ is about $10$ to $20$ times smaller than the vocabulary size, which is also observed for other datasets.

\section{Empirical Scaling Trend} \label{sec:scaling_law_result}

\begin{figure}[t]
    \centering
    \includegraphics[width=0.35\linewidth]{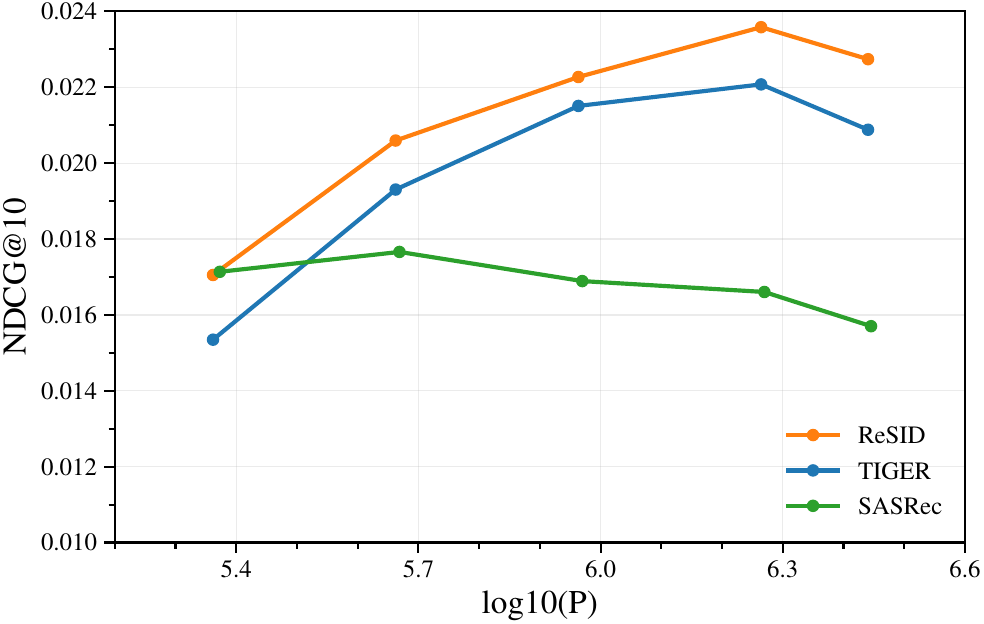}
    \caption{
    Empirical scaling behavior on \textit{Baby Products}.
    The x-axis shows $\log_{10}(P)$, where $P$ is the number of non-embedding model parameters, and the y-axis reports NDCG@10 on the test set.
    For each model size, we select the checkpoint with the best validation NDCG@10 and report its corresponding test NDCG@10.
    We compare ReSID with TIGER and SASRec under matched backbone parameter budgets.
    }
    \label{fig:appx_scaling_ndcg10}
\end{figure}

We investigate how downstream recommendation quality varies with model scale by changing the backbone parameter budget and evaluating ranking metrics. For each configuration, we periodically evaluate on the validation set, select the checkpoint that achieves the best validation NDCG@10, and report its corresponding test NDCG@10 in Fig.~\ref{fig:appx_scaling_ndcg10}. To ensure a fair comparison across different architectures and tokenization schemes, we match models by the number of non-embedding (backbone) parameters.

As shown in Fig.~\ref{fig:appx_scaling_ndcg10}, ReSID consistently achieves the best NDCG@10 across the explored parameter range and exhibits the most favorable scaling behavior, indicating that it leverages additional backbone capacity more effectively under the same data and training protocol. We also observe that performance does not improve monotonically at the largest scale: the final point slightly drops, which is likely due to overfitting and reduced generalization in this low-data regime. Finally, comparing the SID-based models (ReSID and TIGER) with the item-ID-based SASRec baseline under matched backbone budgets suggests that semantic IDs provide a more scalable and parameter-efficient modeling interface for generative recommendation in our setting.

\section{Representation Analysis of FAMAE} \label{sec:embedding_visual}

\subsection{Semantic--Collaborative Alignment in Item Representations}

\begin{figure}[t]
    \centering
    \includegraphics[width=1\linewidth]{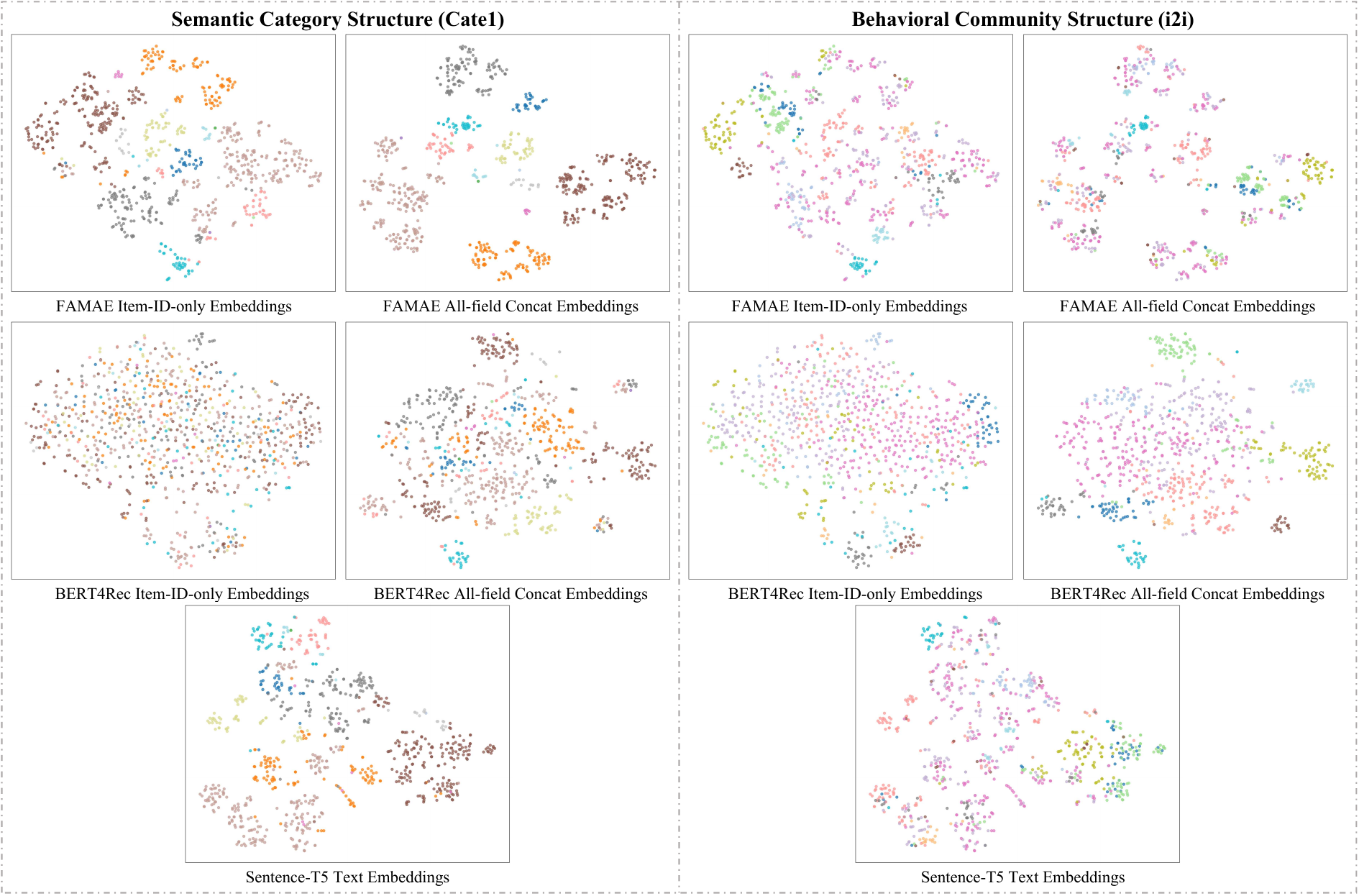}
    \vspace{-2mm}
    \caption{
    T-SNE visualization of item embeddings learned by different methods.
    \textbf{Left:} \emph{Semantic Category Structure (Cate1)}—items are colored by their category labels.
    \textbf{Right:} \emph{Behavioral Community Structure}—items are colored by communities discovered via the Louvain algorithm on a weighted item--item co-occurrence graph constructed from user interaction histories.
    }
    \label{fig:appx_tsne_catei2i_all}
\end{figure}

To examine whether FAMAE captures both semantic and collaborative structures, we visualize item embeddings learned by FAMAE, BERT4Rec, and Sentence-T5 using t-SNE~\cite{tsne}. Items are colored by (i) level-1 category labels to reflect semantic signals, and (ii) behavioral communities to reflect collaborative signals. Behavioral communities are identified using the Louvain algorithm~\cite{louvain} on an item–item co-interaction graph, where edges connect items co-interacted by the same users.

As illustrated in Fig.~\ref{fig:appx_tsne_catei2i_all}, FAMAE exhibits a distinctive advantage over the compared methods. It is the only approach that produces well-structured clusters under both semantic and collaborative views, indicating that its representations simultaneously encode category-level semantics and user interaction-driven relational structure. 

In comparison, embeddings learned by the pre-trained Sentence-T5-xxl model exhibit clear semantic clustering, but are inferior in reflecting behavioral communities (points with the same color are less clustered, e.g., brown and yellow communities scatter and overlap with others), as they are learned purely from textual information without collaborative supervision. Conversely, BERT4Rec embeddings effectively capture collaborative structures but show little semantic organization, resulting in poor category separability.

Moreover, FAMAE’s item-ID embeddings align closely with category-induced semantic structure, whereas BERT4Rec’s item-ID embeddings are largely unstructured. This suggests that FAMAE preserves task-relevant information from structured features, yielding representations that are predictive and sufficient for generative recommendation.

\subsection{Structured and Field-Aligned Embedding Space}

\begin{figure}[t]
    \centering
    \includegraphics[width=0.6\linewidth]{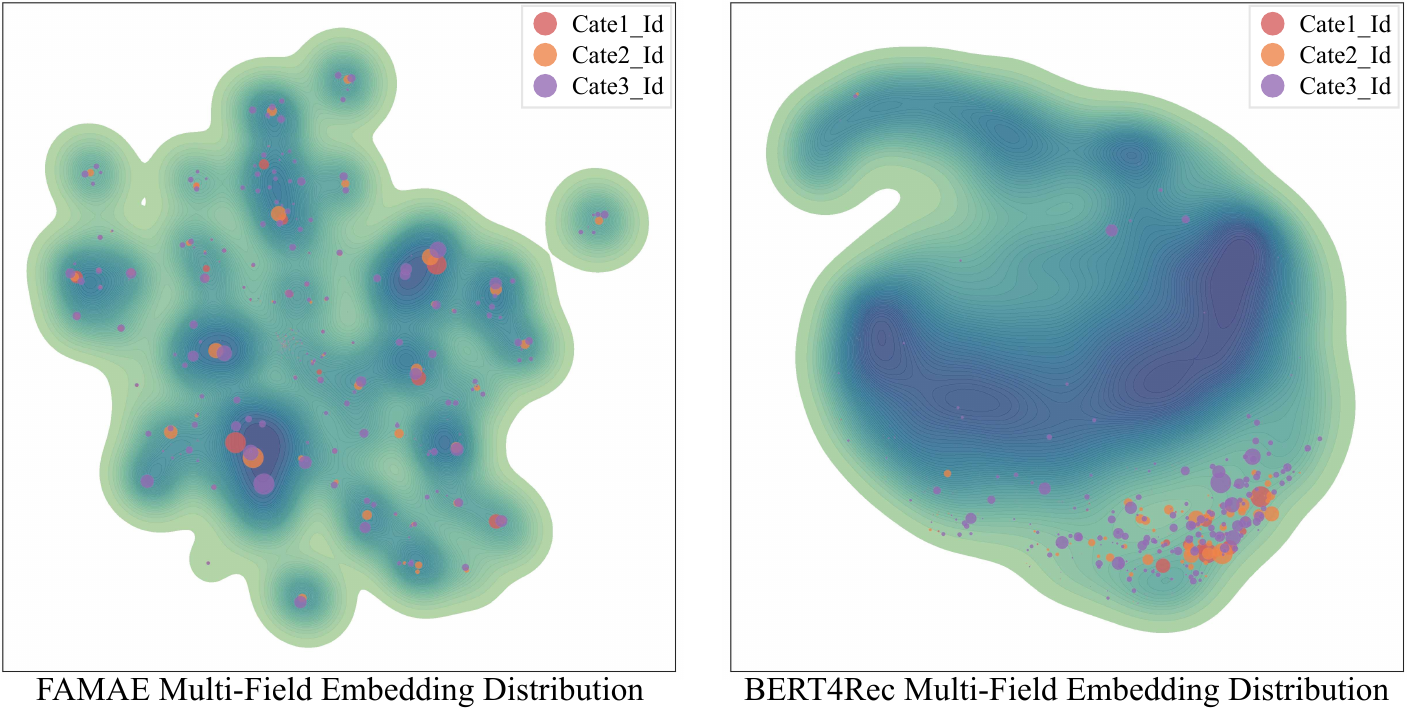}
    \vspace{-2mm}
    \caption{
    UMAP visualization of the joint embedding distribution of item-ID and category fields.
    Item-ID embeddings are rendered as a density map, while category embeddings are displayed as scatter points with sizes proportional to the number of associated items.
    }
    \label{fig:appx_emb_distribution}
\end{figure}

Figure~\ref{fig:appx_emb_distribution} visualizes the joint embedding distribution of item-ID and multi-level category features using UMAP~\cite{umap}.

As shown in the figure, item embeddings learned by FAMAE exhibit a clearly structured and multi-peak distribution, forming multiple localized high-density regions rather than a single unimodal cluster.
This observation suggests that FAMAE learns expressive and discriminative item representations, leading to a more structured embedding geometry.

More importantly, FAMAE yields strong alignment between item-ID embeddings and category embeddings. Category embeddings are distributed around dense item regions, and item density peaks are consistently accompanied by corresponding concentrations of category embeddings.
This correspondence indicates that item identifiers and categorical fields are embedded in a shared and coherent representation space, where category information actively shapes item-level geometry rather than acting as auxiliary side information.

In contrast, embeddings learned by BERT4Rec exhibit a largely unimodal and smooth item distribution with a collapsed internal structure, while category embeddings are heavily concentrated in peripheral regions of the space.

By jointly supervising multiple fields at the final interaction position, FAMAE results in a structured, field-aligned embedding space that provides a favorable foundation for downstream SID construction.

\section{Downstream Task Aligned Design of the FAMAE Objective} \label{sec:task_align_analysis_famae}

\begin{figure}[t]
    \centering
    \includegraphics[width=1\linewidth]{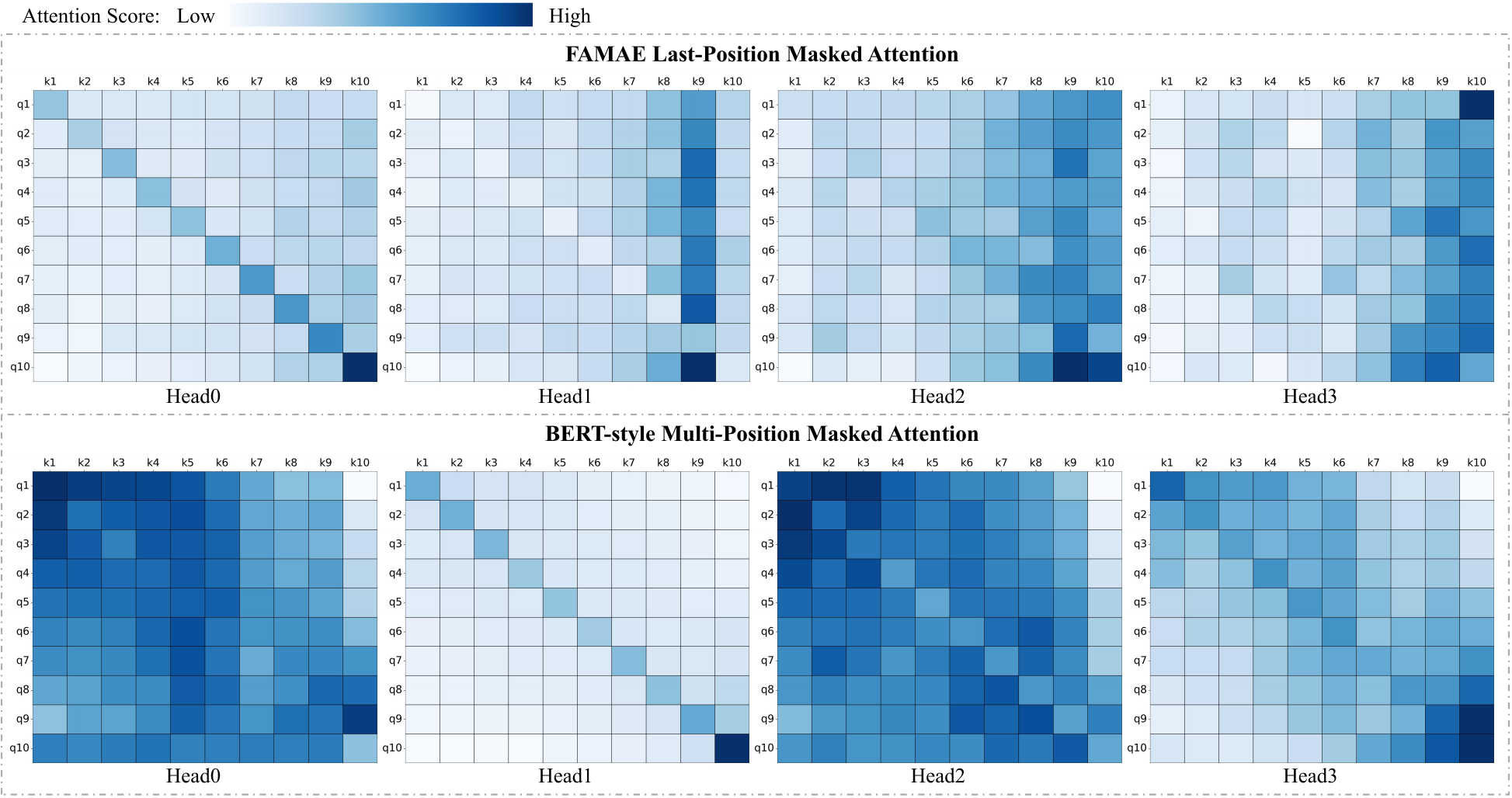}
    \vspace{-2mm}
    \caption{
    Attention score visualizations under different masking-position strategies. The top row uses last-position (target item) masking as used in FAMAE, while the bottom row uses BERT-style masking that randomly masks multiple positions.
    }
    \label{fig:appx_attention}
\end{figure}

\subsection{From the Attention Perspective}
Figure~\ref{fig:appx_attention} compares attention patterns (averaged over validation samples) induced by different masking strategies and illustrates how FAMAE’s last-position, multi-field masking yields task-aligned temporal attention.

Under FAMAE, attention from the target position (lower-right corner) exhibits a clear recency bias, with monotonically increasing weights toward more recent interactions. This reflects the encoder’s focus on aggregating collaborative signals that are most informative for predicting the next item under a sequential recommendation objective. 

In contrast, BERT-style random masking optimizes a position-agnostic reconstruction objective, resulting in diffuse and less structured attention patterns that are not explicitly aligned with next-item prediction.

Overall, these observations show that FAMAE shapes contextual aggregation in a recommendation-native manner, aligning attention structure with the information requirements of downstream generative recommendation rather than generic semantic reconstruction.

\subsection{From the Sequential Decoding Uncertainty Perspective}
\begin{table}[t]
    \caption{
    Overlap ratio between target item codes and historical item codes.
    A historical item is counted as matched if it shares the same code with the target item at the same SID layer.
    Higher overlap indicates stronger task-consistent alignment in the discrete SID space.
    }

    \centering
    \label{tab:appx_sid_overlap}
    
    \begin{tabular}{lcc}
    \toprule
    Method & Code1 & Code2 \\
    \midrule
    FAMAE + GAOQ 
    & \textbf{0.0724} 
    & \textbf{0.0260} \\
    
    FAMAE + Hierarchical K-Means
    & \textbf{0.0724} 
    & 0.0221 \\
    
    Text Embedding + Hierarchical K-Means 
    & 0.0660 
    & 0.0188 \\
    \bottomrule
    \end{tabular}
\end{table}

Table~\ref{tab:appx_sid_overlap} evaluates how effectively task-relevant information encoded in continuous item representations is preserved after discretization into SIDs. Specifically, we measure the overlap between SID tokens of target items and those of their historical interactions, which serves as a proxy for \emph{prefix-consistent relational structure} in the discrete space.

FAMAE-based representations consistently yield higher SID overlap ratios than text-based embeddings, indicating that FAMAE encodes interaction-aligned collaborative structures that survive the discrete bottleneck. From an information-theoretic perspective, this suggests lower reconstruction ambiguity $H(\mathbf{z} \mid c_l)$ and reduced reliance on long prefixes to disambiguate item semantics.

Furthermore, introducing global index alignment in GAOQ—particularly at deeper SID layers—further increases overlap. This demonstrates that GAOQ reduces prefix-dependent ambiguity by enforcing prefix-invariant code semantics, thereby preserving task-aligned structure learned in the E-stage. As a result, the resulting SIDs better reflect user interaction patterns and exhibit lower intrinsic uncertainty for autoregressive decoding.

\section{Understanding GAOQ Algorithm and Advantages}

\begin{figure}[t]
    \centering
    \includegraphics[width=0.9\linewidth]{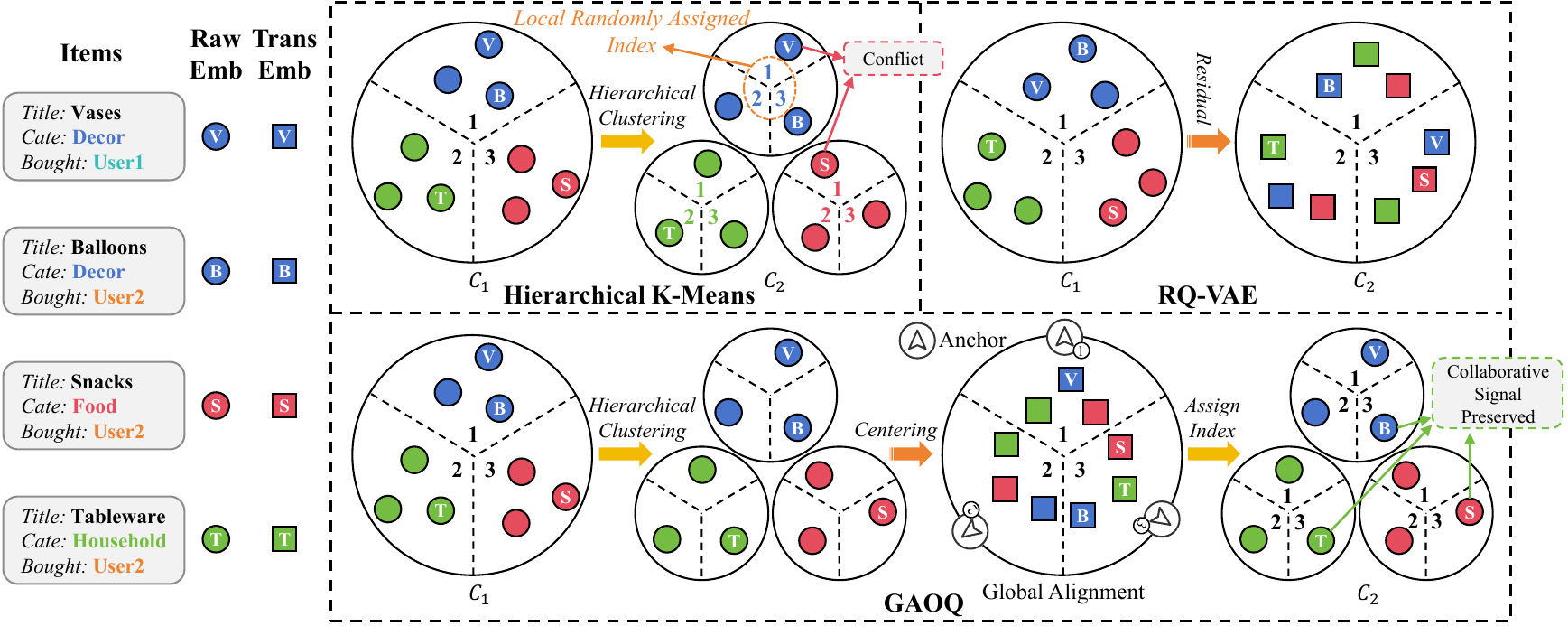}
    \vspace{-2mm}
    \caption{
    Comparison of SID construction strategies.
    \textbf{Hierarchical K-Means} assigns child indices locally and arbitrarily under each parent, with no explicit correspondence across different parents, resulting in non-coherent SIDs.
    \textbf{RQ-VAE} performs residual vector quantization and assigns codes globally and independently across stages, causing individual codes to lack consistent and meaningful semantic interpretation.
    \textbf{GAOQ} introduces a global alignment mechanism on top of Hierarchical K-Means, ensuring that codes have consistent meaning across prefixes while preserving collaborative signals.
    }
    \label{fig:appx_compare}
\end{figure}

\begin{algorithm}[t]
    \small
    \caption{GAOQ at Level $l$ (For a Given Parent at Level $l-1$)}
    \label{alg1}

    \begin{algorithmic}[1]
    \REQUIRE Set of item representations $\mathcal{Z}=\{\mathbf{z}_i\}_{i=1}^{n_p}$ that belong to the current parent node; centroid $\mu$ (the current parent embedding); branching factor $b_l$; number of anchors $g_l$ ($g_l \ge b_l$)
    \ENSURE Level-$l$ code assignment $\mathcal{M}:\mathcal{Z}\rightarrow\{1,\dots,g_l\}$ \COMMENT{mapping each item in $\mathcal{Z}$ to a globally aligned anchor index}
    
    \vspace{0.3em}
    \STATE \COMMENT{Partition (Balanced Clustering)}
    \STATE $(\{\mathcal{Z}_j\}_{j=1}^{b_l}, \{\mu_j\}_{j=1}^{b_l}) \gets \textsc{Balanced-KMeans}(\{\mathbf{z}_i \in\mathcal{Z}\}, b_l)$ 
    \COMMENT{$\{\mathcal{Z}_j\}_{j=1}^{b_l}$ is separated embedding sets of children clusters, $\{\mu_j\}_{j=1}^{b_l}$ is the centroid set of children clusters}
    
    \vspace{0.3em}
    \STATE \COMMENT{Residualization (Centering)}
    \FOR{$j = 1,\dots,b_l$}
        \STATE $\bar{\mu}_j \gets \mu_j - \mu$
    \ENDFOR
    
    \vspace{0.3em}
    \STATE \COMMENT{Anchor Construction}
    \STATE $\mathcal{A}=\{a_k\}_{k=1}^{g_l} \gets \textsc{Ortho}(g_l)$
    \COMMENT{Approximately orthonormal anchor generation function}
    
    \vspace{0.3em}
    \STATE \COMMENT{Global Alignment (Matching)}
    \STATE $W^{b_l\times g_l} \gets \cos(\bar{\mu}_j, a_k)$ for all $(j,k)$ pairs
    \STATE $\mathcal{W} \gets \textsc{Hungarian}(W)$
    \COMMENT{Injective assignment from children centers to anchors}
    
    \vspace{0.3em}
    \STATE \COMMENT{Code Assignment}
    \FOR{$j = 1,\dots,b_l$}
        \FORALL{item $i \in \mathcal{Z}_j$}
            \STATE $\mathcal{M}(i) \gets \mathcal{W}(j)\quad (\in \{1,\dots,g_l\})$ \COMMENT{anchor index lookup for children cluster $j$}
        \ENDFOR
    \ENDFOR
    
    \vspace{0.3em}
    \item[\textbf{RETURN}] $\mathcal{M}$
    \end{algorithmic}
\end{algorithm}

\begin{table}[t]
    \caption{
    Mean pairwise cosine similarity of centered embedding directions for items sharing the same code at second level. 
    Lower values indicate greater directional diversity within the same code.
    }
    \label{tab:appx_code2_cosine}

    \centering
    
    \begin{tabular}{lcccc}
    \toprule
    Method & Musical Instruments & Video Games & Baby Products & Beauty and Personal Care \\
    \midrule
    GAOQ & 0.0463 & 0.0567 & 0.0570 & 0.0524 \\
    Hierarchical K-Means & 0.0172 & 0.0154 & 0.0086 & 0.0053 \\
    \bottomrule
    \end{tabular}
\end{table}

Concretely, at each quantization level, GAOQ first partitions a parent cluster into child clusters using balanced K-Means. It then computes \emph{residual child vectors} by centering each child centroid with respect to its parent centroid. These residuals are matched to a \emph{globally shared} set of anchor embeddings. The anchors are constructed to be approximately orthonormal by maximizing inter-anchor cosine separation, providing a uniform and non-overlapping reference basis. A one-to-one assignment between child clusters and anchors is obtained by maximizing cosine similarity under an injective constraint, solved via the Hungarian algorithm (Algorithm~\ref{alg1}).

Figure~\ref{fig:appx_compare} provides an intuitive comparison (a toy showcase) between GAOQ and prior hierarchical quantization schemes. Consider four items: \emph{Vases} (purchased by User1) and \emph{Balloons}, \emph{Snacks}, \emph{Tableware} (co-purchased by User2). In the embedding space, \emph{Vases} and \emph{Snacks} are far apart due to distinct semantic content, while the latter three items exhibit strong collaborative proximity. 

Under standard Hierarchical K-Means with local indexing, child codes are assigned independently within each parent. As a result, \emph{Vases} and \emph{Snacks} may share the same second-level code (e.g., ``1''), while items co-purchased by User2 receive different codes. This breaks prefix-invariant semantics and causes collaborative structure captured in the embeddings to be discarded during discretization.

RQ-VAE applies residual transformations to the embedding space but assigns codes at each level independently across levels. As a result, codes across different levels are weakly correlated, and higher-level codes do not enforce consistent semantic or collaborative structure across prefixes.

In contrast, GAOQ enforces global alignment across all parent clusters. In the same example, GAOQ assigns a shared code (e.g., ``3'') to all items purchased by User2 and a different code (e.g., ``1'') to \emph{Vases}, preserving collaborative relationships in the resulting SIDs and reducing prefix-dependent ambiguity.

\subsection{Empirical Evidence on Reduced Indexing Ambiguity of GAOQ} \label{sec:GAOQ_lower_ambiguity}
In addition, Table~\ref{tab:appx_code2_cosine} provides evidence of reduced semantic ambiguity among items sharing the same code in GAOQ. Specifically, the average inter-item cosine similarity within each second-level code produced by Hierarchical K-Means is \textbf{3--10$\times$ lower} than that of GAOQ. This indicates that GAOQ yields more directionally coherent code groups, reflecting lower reconstruction ambiguity $H(\mathbf{z} \mid c_l)$ and more prefix-invariant code semantics as our analysis.

\section{Computational Complexity Analysis}
To provide a quantitative understanding of the computational cost of our framework, we analyze the FLOPs of each major stage in the pipeline: \textbf{FAMAE} for representation learning, \textbf{GAOQ} for quantization, and the downstream \textbf{T5}-based generative model. We report \emph{dominant FLOPs} in Big-$O$ form throughout this section, focusing on the leading multiply--accumulate operations and omitting constant factors and lower-order components (e.g., masking, normalization, and loss computation).

\subsection{FLOPs of FAMAE.} \label{sec:famae_flops}
Let $T_e$ denote the input sequence length of FAMAE, $J$ the number of structured features per item, $d_e$ the hidden size of the FAMAE Transformer encoder, and $L_e$ the number of encoder layers. We embed each feature into $\mathbb{R}^{d_e}$.

% before transformer, default: d == d_e (we can remove projection)
% embedding lookup (omit), sum-pooling fusion (T_e * (J - 1) * d), project embedding dim to hidden size (remove), position encoding (T_e * d_e)
% 
% transformer
% layernorm (omit)
% MHSA: QKV projection (T_e * d_e^2), attention score (T_e^2 * d_e), softmax (omit), attention output (T_e^2 * d_e), output projection (T_e * d_e^2), residual (omit), dropout (omit)
% FFN: d_ff = r*d_e, linear (T_e * d_e * d_ff), relu (omit), linear (T_e * d_ff * d_e), residual (omit), layernorm (omit), dropout (omit)
% 
% after transformer
% InfoNCE (omit)

\textbf{Feature fusion.}
At each position, sum-pooling $J$ field embeddings of dimension $d_e$ costs $\mathcal{O}(J d_e)$, and adding positional encoding costs $\mathcal{O}(d_e)$, yielding
\[
\mathcal{O}(T_e J d_e).
\]

\textbf{Transformer encoder.}
In each layer, the dominant FLOPs come from multi-head self-attention (MHSA) and the feed-forward network (FFN). For MHSA, the Q/K/V and output projections cost $\mathcal{O}(T_e d_e^2)$, and the attention score computation ($QK^\top$) and weighted sum ($AV$) cost $\mathcal{O}(T_e^2 d_e)$. For FFN with $d_{\mathrm{ff}} = r d_e$ (constant $r$), the two linear layers cost $\mathcal{O}(T_e d_e d_{\mathrm{ff}}) = \mathcal{O}(T_e d_e^2)$.
Therefore, the dominant FLOPs per encoder layer are
\[
\mathcal{O}(T_e^2 d_e + T_e d_e^2),
\]
and the total FLOPs of an $L_e$-layer encoder are
\[
\mathcal{O}(L_e (T_e^2 d_e + T_e d_e^2)).
\]

\textbf{Overall.}
Combining the above, the dominant FLOPs of FAMAE satisfy
\[
\text{FLOPs}_{\mathrm{FAMAE}} = \mathcal{O}(T_e J d_e + L_e (T_e^2 d_e + T_e d_e^2)).
\]

\subsection{FLOPs of GAOQ.}
Let $N$ be the number of items, $d_q$ the representation dimension fed into GAOQ, and $L_q$ the number of quantization levels. At level $l$ ($l$ is indexed from $1$), GAOQ uses branching factor $b_l$ and $I_l$ iterations for balanced K-Means. For $l \ge 2$, GAOQ additionally uses $g_l$ global anchors for alignment. Let $P_{l}=\prod_{i=1}^{l}b_i$ be the number of nodes at level $l$, and let $n_p$ be the number of items in a node $p$ with $\sum_{p=1}^{P_{l}} n_p = N$.

% before GAOQ, default: d_q == J*d == J*d_e
% get item embedding (omit)
% 
% GAOQ, input: {z_i in R^{d_q}}_{i=1}^N
% level 1 (no global alignment): balanced kmeans on N items
%   I_1 iterations:
%      init kmeans++ (omit), cluster assignment (N * b_1 * d_q), centroid update (N * d_q)
% 
% level >=2: for each parent node p (size n_p), run balanced kmeans + global alignment
%   I_l iterations, aggregated over all parents (sum_p n_p = N):
%     init kmeans++ (omit), cluster assignment: O(N * b_l * d_q), centroid update O(N * d_q)
%   residualization per parent (P_{l-1}):
%       residual (b_l * d_q)
%   anchor construction:
%       QR (g_l^2 * d_q)
%   global alignment per parent (P_{l-1}):
%       cosine sim matrix (b_l * g_l * d_q), hungarian (b_l^3)
%   code assignment in each parent node:
%       code assignment (omit)
% 
% after GAOQ
% get sid (omit)

\textbf{Partition (Balanced Clustering).}
For a given parent node $p$, the dominant cost of balanced K-Means is the assignment step that computes distances between $n_p$ points and $b_l$ centroids, costing $\mathcal{O}(n_p b_l d_q)$ per iteration. Thus, the clustering cost per parent is $\mathcal{O}(I_l n_p b_l d_q)$, and aggregating over all parent nodes of the current level $l$ yields
\[
\mathcal{O}(I_l N b_l d_q).
\]

\textbf{Residualization (Centering).}
For each parent node, we center child representations (sub-cluster centroids) $\bar{\mu}_j = \mu_j - \mu$ for $j=1,\dots,b_l$, which incurs $\mathcal{O}(b_l d_q)$ FLOPs. Since this operation is performed under all $P_{l-1}$ parent nodes at level $l-1$, the per-level cost is
\[
\mathcal{O}(P_{l-1} b_l d_q).
\]

\textbf{Anchor Construction.}
GAOQ constructs $g_l$ approximately orthonormal anchors once per level via QR decomposition on a $d_q \times g_l$ random matrix, costing
\[
\mathcal{O}(d_q g_l^2).
\]

\textbf{Global Alignment (Matching).}
For a parent node $p$, forming the cosine-similarity matrix $W \in \mathbb{R}^{b_l \times g_l}$ costs $\mathcal{O}(b_l g_l d_q)$, and solving the injective assignment via Hungarian costs $\mathcal{O}(b_l^3)$. Aggregating over $P_{l-1}$ parents, the matching cost per level is
\[
\mathcal{O}\big(P_{l-1}(b_l g_l d_q + b_l^3)\big).
\]

\textbf{Code Assignment.}
Assigning the matched anchor index to each item is an index operation and is negligible under our dominant-FLOPs accounting.

\textbf{Per-level and Overall.}
For $l \ge 2$, combining the above stages, the dominant FLOPs at level $l$ satisfy
\[
\mathcal{O}\big(I_l N b_l d_q + P_{l-1} b_l d_q + d_q g_l^2 + P_{l-1}(b_l g_l d_q + b_l^3)\big).
\]
Level $1$ does not use anchors or matching, and only incurs the clustering cost $\mathcal{O}(I_1 N b_1 d_q)$. Therefore, the dominant FLOPs of GAOQ satisfy
\[
\text{FLOPs}_{\mathrm{GAOQ}}
= \mathcal{O}\Big(
\sum_{l=1}^{L_q} I_l N b_l d_q
+
\sum_{l=2}^{L_q}\big[
P_{l-1} b_l d_q + d_q g_l^2 + P_{l-1}(b_l g_l d_q + b_l^3)
\big]
\Big).
\]

\subsection{FLOPs of the Downstream T5 Model.}
Let $T_g^{\mathrm{enc}}$ and $T_g^{\mathrm{dec}}$ denote the encoder input length and decoder output length of the downstream T5 model, respectively. We use a shared hidden size $d_g$ for both encoder and decoder, and set the FFN intermediate dimension as $d_{\mathrm{ff},g} = r_g d_g$ with a constant expansion ratio $r_g$. Let $L_g^{\mathrm{enc}}$ and $L_g^{\mathrm{dec}}$ be the numbers of encoder and decoder layers.

\textbf{Transformer encoder.}
As in Section~\ref{sec:famae_flops}, the dominant FLOPs of a Transformer encoder layer are contributed by multi-head self-attention (MHSA) and the feed-forward network (FFN). For MHSA, the Q/K/V and output projections cost $\mathcal{O}(T_g^{\mathrm{enc}} d_g^2)$, and the attention score computation ($QK^\top$) and weighted sum ($AV$) cost $\mathcal{O}((T_g^{\mathrm{enc}})^2 d_g)$. For FFN, the two linear layers cost $\mathcal{O}(T_g^{\mathrm{enc}} d_g d_{\mathrm{ff},g}) = \mathcal{O}(T_g^{\mathrm{enc}} d_g^2)$. Therefore, the total FLOPs of the $L_g^{\mathrm{enc}}$-layer encoder are
\[
\mathcal{O}\big(L_g^{\mathrm{enc}}\big((T_g^{\mathrm{enc}})^2 d_g + T_g^{\mathrm{enc}} d_g^2\big)\big).
\]

\textbf{Transformer decoder.}
Each decoder layer contains (i) causal self-attention, (ii) encoder--decoder cross-attention, and (iii) an FFN. Causal self-attention has the same dominant FLOPs form as encoder self-attention, yielding $\mathcal{O}((T_g^{\mathrm{dec}})^2 d_g + T_g^{\mathrm{dec}} d_g^2)$ per layer. For cross-attention, projections cost $\mathcal{O}(T_g^{\mathrm{dec}} d_g^2 + T_g^{\mathrm{enc}} d_g^2)$, and the attention score computation and weighted sum cost $\mathcal{O}(T_g^{\mathrm{dec}} T_g^{\mathrm{enc}} d_g)$. The FFN costs $\mathcal{O}(T_g^{\mathrm{dec}} d_g d_{\mathrm{ff},g}) = \mathcal{O}(T_g^{\mathrm{dec}} d_g^2)$. Combining these terms, the total FLOPs of the $L_g^{\mathrm{dec}}$-layer decoder are
\[
\mathcal{O}\Big(L_g^{\mathrm{dec}}\big((T_g^{\mathrm{dec}})^2 d_g + T_g^{\mathrm{dec}} T_g^{\mathrm{enc}} d_g + (T_g^{\mathrm{dec}} + T_g^{\mathrm{enc}}) d_g^2\big)\Big).
\]

\textbf{Overall.}
Combining encoder and decoder, the dominant FLOPs of the downstream T5 model satisfy
\[
\text{FLOPs}_{\mathrm{T5}}
=
\mathcal{O}\Big(
L_g^{\mathrm{enc}}\big((T_g^{\mathrm{enc}})^2 d_g + T_g^{\mathrm{enc}} d_g^2\big)
+
L_g^{\mathrm{dec}}\big((T_g^{\mathrm{dec}})^2 d_g + T_g^{\mathrm{dec}} T_g^{\mathrm{enc}} d_g + (T_g^{\mathrm{dec}} + T_g^{\mathrm{enc}}) d_g^2\big)
\Big).
\]

\subsection{Overall FLOPs.}
The dominant FLOPs of \textbf{ReSID} are the sum of the costs of its three stages:
\[
\text{FLOPs}_{\mathrm{ReSID}} = \text{FLOPs}_{\mathrm{FAMAE}} + \text{FLOPs}_{\mathrm{GAOQ}} + \text{FLOPs}_{\mathrm{T5}}.
\]

\end{document}